\documentclass[letterpaper,11pt]{article}

\usepackage{tikzsymbols}
  
\usepackage[normalem]{ulem} 
\usepackage[11pt]{moresize}
\usepackage{ifthen}

\usepackage{times} 
\usepackage{fullpage} 
\usepackage{amsmath} 
\usepackage{amssymb} 
\usepackage{amsthm} 
\usepackage{amsfonts}
\usepackage{bbm} 
\usepackage{tabularx}
\usepackage{graphicx} 
\usepackage{mathtools}

\usepackage{ifpdf}
\ifpdf
\usepackage[
  pdftex,
  bookmarks,
  plainpages=false, 
  pdfpagelabels=true 
]{hyperref}
\else \fi 
\usepackage{tikz}
\usetikzlibrary{positioning}

\newcommand{\bra}[1]{\langle#1|}

\newcommand{\ket}[1]{|#1\rangle}

\newcommand{\kb}[1]{\ket{#1} \bra{#1}}

\newcommand{\ip}[2]{\langle #1, #2 \rangle} 

\newtheorem{theorem}{Theorem}
\newtheorem{lemma}[theorem]{Lemma}

\theoremstyle{definition}

\newtheorem{defn}[theorem]{Definition}

\newtheorem{protocol}[theorem]{Protocol}

\newcommand{\SFE}{\mathrm{SFE}} 
\newcommand{\ASFE}{\mathrm{A}_{\SFE}} 
\newcommand{\BSFE}{\mathrm{B}_{\SFE}} 
\newcommand{\BSFEP}{\mathrm{B}_{\SFE}'} 
\newcommand{\ASFER}{\mathrm{A}_{\mathrm{rand}}}
\newcommand{\BSFER}{\mathrm{B}_{\mathrm{rand}}'}   
 
\newcommand{\OT}{\mathrm{OT}} 
\newcommand{\AOT}{\mathrm{A}_{\OT}} 
\newcommand{\BOT}{\mathrm{B}_{\OT}}  
\newcommand{\XOT}{\mathrm{XOT}} 
\newcommand{\IP}{\mathrm{IP}} 
\newcommand{\MP}{\mathrm{\Smiley \$}}  
\newcommand{\GS}{\mathrm{\Sey \$}}  
\newcommand{\AXOT}{\mathrm{A}_{\XOT}} 
\newcommand{\BXOT}{\mathrm{B}_{\XOT}}  
\newcommand{\DR}{\mathrm{DR,n}}
\newcommand{\ADR}{\mathrm{A}_{\DR}}
\newcommand{\BDR}{\mathrm{B}_{\DR}}
\newcommand{\KNOT}{\mathrm{knOT}} 
\newcommand{\AKNOT}{\mathrm{A}_{\KNOT}} 
\newcommand{\BKNOT}{\mathrm{B}_{\KNOT}} 
\newcommand{\EQ}{\mathrm{EQ}} 
\newcommand{\AEQ}{\mathrm{A}_{\EQ}} 
\newcommand{\BEQ}{\mathrm{B}_{\EQ}} 
\newcommand{\AIP}{\mathrm{A}_{\IP}} 
\newcommand{\BIP}{\mathrm{B}_{\IP}} 
\newcommand{\AMP}{\mathrm{A}_{\MP}} 
\newcommand{\BMP}{\mathrm{B}_{\MP}} 
\newcommand{\AGS}{\mathrm{A}_{\GS}} 
\newcommand{\BGS}{\mathrm{B}_{\GS}}

\usepackage{pgfplots}
\usepgfplotslibrary{colorbrewer}
\pgfplotsset{compat=newest}
\pgfplotsset{cycle list/Set2}

\pgfplotsset{IneqStyleGtr/.style = {%
   decoration={border,segment length=1mm, amplitude=3mm,angle=90},
              postaction={decorate,draw},
              draw opacity=0.5
  }
}

\begin{document}

\title{
A constant lower bound for any quantum protocol for secure function evaluation}

\author{
  Sarah Osborn$^*$ \qquad Jamie Sikora$^\dagger$  \\[4mm]
  {\small\it
  \begin{tabular}{c}
   {\large$^*$}Virginia Polytechnic Institute and State University - \tt{sosborn@vt.edu} \\[1mm] 
   {\large$^\dagger$}Virginia Polytechnic Institute and State University - \tt{sikora@vt.edu} \\[1mm]
  \end{tabular}
  }
}

\date{March 16, 2022} 

\maketitle 
 
\begin{abstract}  
Secure function evaluation is a two-party cryptographic primitive where Bob computes a function of Alice's and his respective inputs, and both hope to keep their inputs private from the other party.
It has been proven that perfect (or near perfect) security is impossible, even for quantum protocols. 
We generalize this no-go result by exhibiting a constant lower bound on the cheating probabilities for any quantum protocol for secure function evaluation, and present many applications from oblivious transfer to the millionaire's problem. Constant lower bounds are of practical interest since they imply the impossibility to arbitrarily amplify the security of quantum protocols by any means. 
\end{abstract} 


\section{Introduction}
The first paper studying quantum cryptography was written by Stephen Wiesner in the 1970s (published in 1983)~\cite{W83}. 
In that paper, he presented a (knowingly insecure) protocol for \emph{multiplexing} where a receiver could choose to learn one of two bits of their choosing. 
Since then, this task has been referred to as \emph{$1$-out-of-$2$} oblivious transfer, and has been extensively studied in the quantum community~\cite{DFSS05, C07, WST08, STW09, L97, BCS12, CKS13, CGS16, GRS18, KST20}.   
Indeed, since the development of quantum key distribution in 1984~\cite{BB84}, it has been of great interest to use quantum mechanics to develop protocols for classical tasks and push the limits of quantum theory to find optimal protocols (and their limitations).  

On the other hand, it was shown in the late 1990s (and a few times since) that perfect security for a number of cryptographic tasks, including secure function evaluation, could not attain perfect, or even near perfect, security~\cite{M97,LC97,LC98,L97,BCS12}. 
Indeed, some popular two-party cryptographic protocols, including bit commitment~\cite{CK11}, strong coin flipping~\cite{K02}, oblivious transfer~\cite{CKS13}, strong die rolling~\cite{AS10}, as well as many others, have all seen constant lower bounds presented. 
Constant lower bounds are of great interest for several reasons, of which we note a few. 
The first reason, a practical one, is that they imply that there is no way to arbitrarily amplify the security by any means (such as repeating the protocol many times and combining them in some way). 
The second reason, a theoretical one, now opens the question as to what are the optimal security parameters. 
Assuming quantum mechanics offers \emph{some} advantage over their classical counterparts, the question now becomes to what extent is this advantage.

Note that two-party cryptography has some strange behavior, making its study very intriguing. For example, in the case of die rolling (where Alice and Bob wish to roll a die over the (possibly quantum) telephone) there can sometimes be classical protocols that offer decent security~\cite{S17}. 
On the other hand, having classical protocols for something like coin flipping, bit commitment, and oblivious transfer is impossible~\cite{K02}. 
And while quantum mechanics seems to deny us strong coin flipping (we have a constant lower bound~\cite{K02}), it does give us arbitrarily good security for weak coin flipping~\cite{M07, ACGKM16}.
Therefore, classifying the behavior of two-party cryptographic primitives is a fruitful, and sometimes surprising, endeavor. 
To this end, we study the broad class of two-party cryptography known as secure function evaluation which we now discuss.  


\subsection{Secure function evaluation} 
\label{intro_SFE}
  
\emph{Secure function evaluation} (SFE) is a two-party cryptographic primitive in which Alice begins with an input $x \in X$ and Bob begins with an input $y \in Y$ (each input is chosen uniformly at random\footnote{We believe our analysis works for other probability distributions over the inputs as well, as long as they are uncorrelated. 
The assumption of uniformity makes certain expressions cleaner, such as the probability of Alice being able to blindly guess Bob's input.}) and Bob has a deterministic function $f : X \times Y \to B$. 
Here, we take $X$, $Y$, and $B$ to have finite cardinality. See Figure~\ref{fig_sfe} below. 

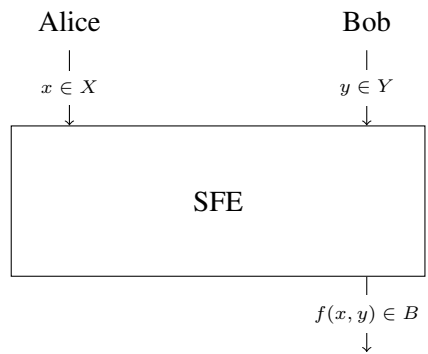
\begin{figure}[h]
\centering
\begin{tikzpicture}[node distance=1.11cm]
\node[draw,minimum width=5.5cm, minimum height=2cm] (SFE) {SFE};
\node (Alice) [above=of SFE.153] {Alice};
\node (Bob) [above=of SFE.27] {Bob};
\draw[->] ([yshift=1cm]SFE.153) --  node[fill=white] {\ssmall $x \in X$} +(0pt,-1cm);
\draw[->] ([yshift=1cm]SFE.27) -- node[fill=white] {\ssmall $y \in Y$} +(0pt,-1cm);
\draw[->] (SFE.333) -- node[fill=white] {\ssmall $f(x,y) \in B$} +(0pt,-1cm); 
\end{tikzpicture}
\caption{A pictorial representation of SFE. Bob wants to compute his function $f$ while he and Alice keep their inputs private.} 
\label{fig_sfe}
\end{figure}

The goals when designing a (quantum) protocol for SFE are: 
\begin{enumerate}
  \item \emph{Completeness:} If both parties are honest then Bob learns $f$, evaluated on their inputs $x$ and $y$.
  \item \emph{Soundness against cheating Bob:} Cheating Bob obtains no extra information about honest Alice's input $x$ other than what is logically implied from knowing $f(x,y)$.
  \item \emph{Soundness against cheating Alice:} Cheating Alice obtains no information about honest Bob's input $y$.
\end{enumerate} 

It is natural to assume perfect completeness of a protocol and then to quantify the extent to which they can be made sound. 
In other words, we consider protocols for SFE that do what they are meant to do when Alice and Bob follow them (that is, they compute $f$), and then we try to find the ones which hide their respective inputs the best. 

To quantify soundness against cheating Bob, for each such protocol we define the following symbols. 
\begin{center}
\begin{tabularx}{\textwidth}{rX}
  $\BSFE$: & The maximum probability with which cheating Bob can guess honest Alice's input $x$. \\ 
  $\BSFEP$: & The maximum probability with which cheating Bob can guess every $f(x,y)$,  \emph{for each} $y \in Y$. \\ 
\end{tabularx}
\end{center}

Note that often these two \emph{cheating probabilities} are the same. 
For instance, in $1$-out-of-$2$ oblivious transfer we have $x \in X = \{ 0, 1 \}^2$ as a $2$-bit string, $y \in Y = \{ 1, 2 \}$ as an index, and $f(x, y) = x_y$, i.e., the $y$-th bit of $x$. 
Then clearly $\BSFE = \BSFEP$, since knowing each bit is equivalent to knowing the full string.  
In general, $\BSFE \leq \BSFEP$, since if Bob is able to correctly learn Alice's input $x$, then he can compute any function of it he wants. 

Similarly, to quantify soundness against cheating Alice, we define the following symbols.  
\begin{center}
\begin{tabularx}{\textwidth}{rX} 
  $\ASFE$: & The maximum probability with which cheating Alice can guess honest Bob's input $y$. \\ 
\end{tabularx}
\end{center} 

Note that there is only the one definition for a cheating probability for Alice since she has no output. 
     
    
\subsection{Main result} \label{intro_Main}  
  
We now present our main result, a trade-off curve relating Alice and Bob's cheating probabilities that must be satisfied for any quantum protocol for SFE.  

\begin{theorem} \label{thm:LowerBound}
In any quantum protocol for secure function evaluation, it holds that 
\begin{equation} \label{lb}
\BSFEP \geq \frac{1}{\left|Y\right|\ASFE}-2\left(\left|Y\right|-1\right)\sqrt{1-\frac{1}{\left|Y\right|\ASFE}} 
\end{equation} 
where $Y$ is the set of choices for Bob's input. 
\end{theorem}  

We now discuss this bound. 
Note that 
\begin{equation}
\ASFE \geq \frac{1}{|Y|}, 
\end{equation} 
since she can always blindly, or randomly, guess the value of $y \in Y$. 
Since Alice has no output function (like Bob does) she may not be able to infer anything about $y$ from the protocol if she is honest. 
Therefore, sometimes her best strategy is to randomly guess, and in this case we would have 
\begin{equation}
\ASFE = \frac{1}{|Y|}, 
\end{equation} 
which translates to perfect security against a cheating Alice. 
However, in that case, our bound implies that 
\begin{equation}
\BSFEP = 1, 
\end{equation} 
meaning Bob can compute his function perfectly for \emph{every choice of input on his side}, i.e., complete insecurity against a cheating Bob.  
This implication exactly recovers Lo's conclusion in his 1997 paper~\cite{L97}, and also the conclusion in a more recent paper by Buhrman, Christandl, and Schaffner~\cite{BCS12}. 
It should be mentioned that the above two papers also discuss the ``Alice can cheat with a small probability" case as well. 
A key component in their proofs is the application of Uhlmann's theorem on purifications of the protocol to find unitaries with which Bob can use to cheat. 
As evidenced later on, this is very different from our proof. 
In fact, at no point in our protocol do we assume anything is pure and we only deal with POVMs, not unitaries. 
The ``magic ingredient'' in our proof is a generalization of Kitaev's lower bound for strong coin flipping~\cite{K02}. 
Moreover, we chose to quantify the security solely in terms of Alice and Bob's cheating probabilities, which is complementary to the results in~\cite{BCS12}. 
  
Before continuing, we now discuss what we mean by having a ``constant lower bound.'' 
To this end, we define the following symbols. 
\begin{center}
\begin{tabularx}{\textwidth}{rX}
  $\ASFER$: & The maximum probability with which cheating Alice can guess honest Bob's input $y$ \emph{given only black-box access to the SFE task}. \\ 
  $\BSFER$: & The maximum probability with which cheating Bob can learn every $f(x,y)$,  \emph{for each} $y \in Y$, \emph{given only black-box access to the SFE task}. \\ 
\end{tabularx} 
\end{center}  

In other words, the cheating definitions above correspond to the information Alice and Bob can infer only from their outputs. 
Of course, Alice has no output, so clearly 
\begin{equation}
\ASFER = \frac{1}{|Y|}.  
\end{equation}  
However, as is illustrated in our examples, it is less clear how to write $\BSFER$ in terms of the parameters of a general SFE protocol. 

Equipped with these symbols, we are now ready to state our constant lower bound on SFE. 

\begin{theorem} 
In any quantum protocol for secure function evaluation, either $\BSFER = 1$ (in which case the protocol is \emph{completely} insecure), or there exists a constant $c > 1$ such that 
\begin{equation} 
\ASFE \geq c \cdot \ASFER 
\qquad \text{ or } \qquad  
\BSFEP \geq c \cdot \BSFER. 
\end{equation} 
\end{theorem} 

Before discussing how to find this constant, a word on our lower bound is in order. 
We chose to define what it means for a constant lower bound to be a multiplicative factor. 
This is because $\ASFER$ and $\BSFER$ may be dramatically different (as we demonstrate shortly). 
Therefore, having a constant additive factor could be unevenly weighted between cheating Alice and Bob and, we feel, would be less insightful in those cases. 
However, using our bound one can optimize and find an additive constant if one so desires. 

To find this constant $c > 1$, note that our lower bound on $\BSFEP$ (the right-hand side of Inequality~\eqref{lb}) is a continuous, decreasing function with respect to $\ASFE$. 
Therefore, if we assume 
\begin{equation}\label{ASFEc}
\ASFE \leq \frac{c_A}{|Y|},  
\end{equation} 
for some fixed constant $c_A \geq 1$, then we may conclude via our bound that 
\begin{equation} \label{lb2}
\BSFEP \geq \frac{1}{c_A}-2\left(\left|Y\right|-1\right)\sqrt{1-\frac{1}{c_A}}.  
\end{equation}   
Now, assuming that 
\begin{equation}\label{BSFEc}
\BSFEP = c_B \cdot \BSFER 
\end{equation}  
for some $c_B \geq 1$, we now have the inequality  
\begin{equation} \label{lb3}
c_B \geq \frac{1}{\BSFER} \left( \frac{1}{c_A}-2\left(\left|Y\right|-1\right)\sqrt{1-\frac{1}{c_A}} \right).   
\end{equation}   

We shall now assume that $\BSFER < 1$ so that $\frac{1}{\BSFER} > 1$. 
Note that when $c_A = 1$, we have the right-hand side of \eqref{lb3} equalling $\frac{1}{\BSFER} > 1$ and when $c_A = \frac{1}{\BSFER}$ we have the right-hand side being strictly less than $1$. 
Thus, by continuity of the right-hand side and the intermediate value theorem, we know there exists a constant $c > 1$ satisfying the equation  
 \begin{equation} \label{exact}
c = \frac{1}{\BSFER} \left( \frac{1}{c}-2\left(\left|Y\right|-1\right)\sqrt{1-\frac{1}{c}} \right).   
\end{equation}   
Note that this constant $c > 1$ is exactly what we want, since if $c_A \leq c$ then we have $c_B \geq c$. 

Now, in theory one can solve for $c$ above for a general SFE task, but it is complicated and perhaps not very insightful. 
However, when it comes to particular instances or families of SFE, then one can easily solve the above equation and get a constant (and possibly decent) lower bound for any quantum protocol for that task. 
We demonstrate this several times below.  
               
    
\subsection{Applications} \label{intro_App}  

Since our bound is general, we can apply it to many different scenarios. 
However, since each scenario is quite different and requires discussion, we delegate these discussions to their own section and simply summarize the cryptographic tasks below and a few of the special cases in which we found some exact formulas for lower bounds. 
Note that all of the special cases presented below are \emph{new lower bounds} as far as we are aware. 

\begin{itemize} 
\item \textbf{$1$-out-of-$n$ oblivious transfer} (Subsection~\ref{OT}). 
This is where Alice has a database and Bob wishes to learn one item (his input is an index). 
We present lower bounds on either how much Alice can learn Bob's index or how much Bob can learn all of Alice's database. 
A special case is when Alice has $3$ bits and Bob wants to learn $1$ of them. 
We present a new lower bound that either 
\begin{equation} 
\BOT \gtrapprox 0.2581 > 0.2500 
\quad \text{ or } \quad 
\AOT \gtrapprox 0.3442 > 0.3333.
\end{equation} 
Note that we define the cheating probability symbols above in Section~\ref{sect:app}, but they should be clear from context for this abbreviated discussion. This is also the case for the cheating probability symbols below.


\item \textbf{$k$-out-of-$n$ oblivious transfer}  (Subsection~\ref{KNOT}). 
This is the same as $1$-out-of-$n$ oblivious transfer except Bob's input is now a proper subset instead of an index (so Bob learns $k < n$ entries in Alice's database). 
We present lower bounds on either how much Alice can learn Bob's proper subset or how much Bob can learn all of Alice's database. 
A special case is when Alice has $4$ bits and Bob wants to learn $2$ of them. 
We present a new lower bound that either 
\begin{equation} 
\BKNOT \gtrapprox 0.2514 > 0.2500 
\quad \text{ or } \quad 
\AKNOT \gtrapprox 0.1676 > 0.1667.
\end{equation} 


\item \textbf{XOR oblivious transfer}  (Subsection~\ref{XOT}). 
This is similar to $1$-out-of-$2$ oblivious transfer (where Alice's database consists of $2$ bit \emph{strings}) but Bob now has a third option of learning the bit-wise XOR of the two strings. 
We present lower bounds on either how much Alice can learn Bob's choice (first string, second string, or the XOR) or how much Bob can learn both of Alice's strings. 
A special case is when Alice's strings have length $1$ (so, they are just bits). 
We present a new lower bound that either 
\begin{equation} 
\BXOT \gtrapprox 0.5073 > 0.5000
\quad \text{ or } \quad 
\AXOT \gtrapprox 0.3382 > 0.3333.
\end{equation}


\item \textbf{Equality/{one-way oblivious identification}}   (Subsection~\ref{EQ}). 
This is when Alice and Bob each have the same set of inputs and Bob learns whether or not their inputs are equal. 
We present lower bounds on either how much Alice or Bob can learn the other's input. 
A special case is when the input set has cardinality $3$.  
We present a new lower bound that either 
\begin{equation} 
\BEQ \gtrapprox 0.671 > 0.667
\quad \text{ or } \quad 
\AEQ \gtrapprox 0.3355 > 0.3333.
\end{equation}


\item \textbf{Inner product}  (Subsection~\ref{IP}). 
This is when Alice and Bob each input an $n$-bit string and Bob learns their inner product. 
We present lower bounds on either how much Alice or Bob can learn the other's input. 
A special case is when $n=3$.   
We present a new lower bound that either 
\begin{equation} 
\BIP \gtrapprox 0.251 > 0.250 
\quad \text{ or } \quad 
\AIP \gtrapprox 0.1434 > 0.1429.
\end{equation}


\item \textbf{Millionaire's problem}  (Subsection~\ref{MP}). 
This is when (rich) Alice and Bob have lots of money and Bob wishes to learn who is richer without either revealing their wealth. 
A special case is when $n=10^9$ (bounding each of their bank accounts at a billion dollars).   
We present a new lower bound that either 
\begin{equation} 
\begin{split}
\BMP &\gtrapprox 2 \times 10^{-9} + 5 \times 10^{-28}  > 2 \times 10^{-9}
\quad \text{ or } \quad \\
\AMP &\gtrapprox 1 \times 10^{-9} + 1 \times 10^{-18} + 1.25 \times 10^{-27} > 1 \times 10^{-9} + 1 \times 10^{-18} + 1 \times 10^{-27}.
\end{split}
\end{equation} 
We can also study the version geared towards academics by setting $n = 10$. The substitution of sad faces indicates Bob and Alice's attitudes towards their financial situations.
We present a new lower bound that either 
\begin{equation} 
\BGS \gtrapprox 0.2005 > 0.2000
\quad \text{ or } \quad 
\AGS \gtrapprox 0.1114 > 0.1111.
\end{equation} 
Therefore, some information about either Alice or Bob's wealth is necessarily leaked. 
\end{itemize} 
   
Each of these cryptographic tasks are described further and analyzed in Section~\ref{sect:app}. 

 
\subsection{Proof idea and key concepts} 

There are two main ingredients in proving our lower bound which we discuss at a high level below, and continue in more detail in the following sections. 
The magic ingredient is Kitaev's constant lower bound for die rolling~\cite{K02, AS10}. 
Effectively what we do is use a generic SFE protocol to create a die rolling protocol, then apply Kitaev's lower bound. 
However, the glue that makes SFE and die rolling play well together is a new technical result that we prove which deals with sequential gentle measurements, which we discuss next.


\subsubsection{Sequential gentle measurements}

The idea behind much of quantum cryptography is the concept of measurement disturbance. 
To put it simply, measuring to obtain certain information from a quantum state may cause it to collapse, possibly rendering it unusable for future purposes, or to simply alert honest parties that a cheating attempt was made. 
However, there is a concept of a \emph{gentle measurement}, which is described at a high level below. 

\medskip 
\noindent 
\textbf{Gentle measurement lemma ($\epsilon$-free version):} \textit{If a measurement outcome has a large probability of occurring, then the measured quantum state is not largely disturbed if that measurement outcome does indeed happen. (See references~\cite{W19, W99} or Section~\ref{lgm section} to see formal statements of gentle measurement lemmas.)} 
\medskip 

How does this help us? Well, suppose for a cheating Bob who wishes to learn every $f(x,y)$, for all $y$, he may wish to measure some quantum state \emph{several times}. 
Suppose for a fixed $y_1 \in Y$ that Bob can learn $f(x, y_1)$ with probability close to $1$. 
Then, if he were to measure it, and achieve the correct value, then the state is not greatly disturbed, and thus more information can possibly be extracted. 
If a second measurement can extract the correct value of $f(x, y_2)$ for some $y_2 \in Y \setminus \{ y_1 \}$ with a high probability, we can repeat the process. 

Now, we (intentionally) glossed over the concept of \emph{learning}, that is, we did not precisely define it means to \emph{learn} the correct value, in our cryptographic context.  
We elaborate on this in Section~\ref{lgm section}. 
However, it can be made precise and be put into a framework suitable for the application of a modified gentle measurement lemma. 
For now, we just state the main technical result of this paper below, and leave its proof for Subsection~\ref{encodings}. 

\begin{lemma}[Sequential measurement lemma] \label{lem:SGMLformal}
Let $f_1, \ldots, f_n : X \to B$ be fixed functions and suppose Bob has a quantum encoding of $x \in X$ (where $x$ is chosen from a probability distribution known to Bob). 
Suppose Bob can learn $f_i(x)$ with probability $p_i$ for each $i \in \{ 1, \ldots, n \}$ and let $p = \frac{1}{n} \sum_{i=1}^n p_i$ be his average success probability of learning the function values. 
Then Bob can learn \emph{all} values $f_1(x), \ldots, f_n(x)$ with probability at least 
\begin{equation} \label{SGMLinformal}
p - 2(n-1)\sqrt{1-p}. 
\end{equation} 
\end{lemma}  

Notice that if $p \approx 1$ (meaning that Bob has a high average success probability of learning the function values) then he can learn \emph{all} the values with probability still very close to $1$. 
Note that this aligns with the intuition one obtains from the gentle measurement lemma. 
The measurement that achieves the success probability in Lemma~\ref{lem:SGMLformal} is given in Subsection~\ref{encodings}. 


\subsubsection{Die rolling}
\emph{Die rolling} (DR) is a two-party cryptographic task akin to coin flipping, where Alice and Bob try to agree on a value $n \in \{ 0, 1, \ldots, N-1 \}$. 
The goals when designing a die rolling protocol are outlined below. 
\begin{enumerate}
  \item \emph{Completeness:} If both parties are honest then their outcomes are uniformly random and identical.
  \item \emph{Soundness against cheating Bob:} Cheating Bob cannot influence honest Alice's outcome distribution away from uniform. 
  \item \emph{Soundness against cheating Alice:} Cheating Alice cannot influence honest Bob's outcome distribution away from uniform. 
\end{enumerate} 

For this work, we only consider perfectly complete die rolling protocols. 
To quantify the soundness of a die rolling protocol, we define the following symbols. 
\begin{center}
\begin{tabularx}{\textwidth}{rX}
  $\BDR$: & The maximum probability with which cheating Bob can influence honest Alice to output the number $n$ without Alice aborting. \\
  $\ADR$: & The maximum probability with which cheating Alice can force honest Bob to output the number $n$ without Bob aborting. \\ 
\end{tabularx}
\end{center}

Kitaev proved in~\cite{K02} that when $N = 2$,  any \emph{quantum} protocol for die rolling satisfies 
\begin{equation}
\mathrm{A}_{\mathrm{DR,0}}\mathrm{B}_{\mathrm{DR,0}} \geq \frac{1}{2}
\quad \text{ and } \quad 
\mathrm{A}_{\mathrm{DR,1}}\mathrm{B}_{\mathrm{DR,1}} \geq \frac{1}{2}.  
\end{equation}  
Note that die rolling with $N = 2$ is simply referred to as \emph{(strong) coin flipping} as Alice and Bob decide on one of two outcomes. 
Note that coin flipping is a much more studied task than die rolling, the latter being a generalization of the former.  
Kitaev's proof of these inequalities for coin flipping easily generalizes to similar inequalities for die rolling, namely that for any \emph{quantum} protocol for die rolling, we have 
\begin{equation} \label{KitDR1}
\ADR \BDR \geq \frac{1}{N}, \; \text{ for all } \; n \in \{ 0, 1, \ldots, N-1 \}. 
\end{equation} 
This is indeed a constant lower bound, as we would strive to have $\ADR = \BDR = \frac{1}{N}$ for all $n$. 
However, Inequality~\eqref{KitDR1} implies that 
\begin{equation} \label{KitDR2}
\max \{ \ADR, \BDR \} \geq \frac{1}{\sqrt N}, \; \text{ for all } \; n \in \{ 0, 1, \ldots, N-1 \} 
\end{equation}  
making it impossible to get anywhere near perfect security. 


\subsubsection{Die rolling via secure function evaluation - gluing the two ingredients together}

The first step is to create a DR protocol from a fixed SFE protocol, as shown below.  

\begin{protocol}[DR from SFE] \label{DRprotocol}
\quad 
\begin{itemize}
\item Alice and Bob input uniformly random chosen inputs into a SFE protocol such that Bob learns $f(x,y)$. 
\item Alice selects a uniformly chosen $b \in Y$, independent from the SFE protocol. 
She sends $b$ to Bob.  
\item Bob reveals his SFE input $y \in Y$ and also his SFE output $f(x,y)$. 
\item Alice computes $f(x,y)$ using $x$ and $y$. 
If Bob's function value he sent to Alice does not match Alice's computation of the function, she aborts the protocol. 
\item If Alice does not abort, they both output $(b + y) \mod |Y|$. We assume an ordering of the elements of $Y$ is known to both Alice and Bob before the protocol, i.e., we may think of them as elements of the set $\{ 1, \ldots, |Y| \}$.
\end{itemize}
\end{protocol}  

Protocol~\ref{DRprotocol} is pictorially shown in Figure~\ref{fig_sfe_dr1} below. 
  
\begin{figure}[h] 
\centering
\begin{tikzpicture}[node distance=0.2cm]
\node[draw, minimum width=6cm, minimum height=5.75cm] (Outer_DR) {};
\node[draw,minimum width=5.5cm, minimum height=2cm] at (0,0.875) (SFE_DR) {SFE};
\node (Alice) [above=of Outer_DR.121] {Alice};
\node (Bob) [above=of Outer_DR.56] {Bob};
\draw[->] ([yshift=1cm]SFE_DR.150) --  node[fill=white] {\ssmall $x \in X$} +(0pt,-1cm);
\draw[->] ([yshift=1cm]SFE_DR.27) -- node[fill=white] {\ssmall $y \in Y$} +(0pt,-1cm);
\draw[->] (SFE_DR.333) -- node[fill=white] {\ssmall $f(x,y) \in B$} +(0pt,-1cm); 
\draw[<-] (2.75,-1.5)--node[fill=white] {\ssmall $b \in Y$}+(-5.5,0);
\draw[->] (2.75,-2)--node[fill=white] {\ssmall $y\text{, }f(x,y)$}+(-5.5,0);
\draw[->] (Outer_DR.239) -- node[fill=white] {\ssmall $(b+y)\text{ mod }\left|Y\right|$} +(0pt,-1cm);
\draw[->] (Outer_DR.304) -- node[fill=white] {\ssmall $(b+y)\text{ mod }\left|Y\right|$} +(0pt,-1cm);
\node at (-1.7,-2.5) (AC) {\ssmall Alice checks $f(x,y)$};
\end{tikzpicture}
\caption{Protocol~\ref{DRprotocol}: Die rolling via secure function evaluation.} 
\label{fig_sfe_dr1}
\end{figure}
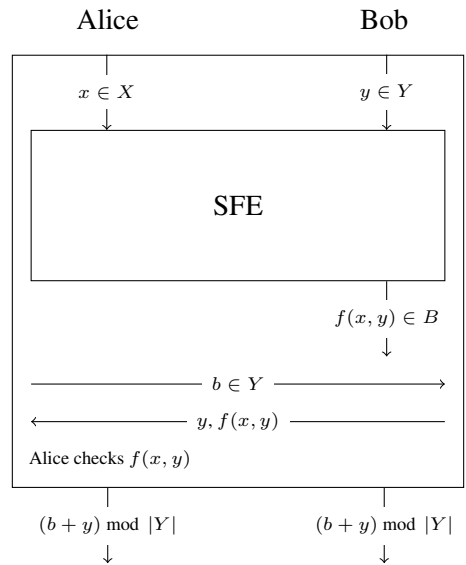  

We now describe what Alice and Bob may do to cheat the die rolling protocol. 
 
\paragraph{Cheating Alice.} 
Suppose cheating Alice wants to force honest Bob to output the number $0$. 
In this case, Alice must send $b$ in the second to last message such that $b = y$. 
Since she may not know $y$, the probability she can successfully cheat is equal to the maximum probability with which she can learn $y$ from the SFE protocol. 
However, this is precisely the definition of $\ASFE$. 
Thus, the case of cheating Alice is simple, we have that $\mathrm{A}_{\mathrm{DR,0}} =  \ASFE$.  

\paragraph{Cheating Bob.} 
Similar to cheating Alice, we wish to relate how much Bob can cheat in the DR protocol, say the quantity $\mathrm{B}_{\mathrm{DR,0}}$, and how much he can cheat in the SFE protocol, namely $\BSFEP$. 
Suppose cheating Bob wants to force an honest Alice to output the number $0$. 
In this case, he needs to send back $y$ such that $y = b$ in the last message. 
However, for Alice to accept this last message, he must also correctly \emph{learn} the value $f(x,y)$ from his part of the state after the SFE subroutine. 
In other words, before he sends his last message, he has an \emph{encoding} of $x$ from which he may measure to learn something.  
Since Alice's message $b$ is randomly chosen, independent of the SFE protocol, he is tasked with revealing a $y$ with uniform probability. 
To say it another way, $\mathrm{B}_{\mathrm{DR,0}}$ is equal to the average probability that Bob is able to learn $f(x,y)$, for each $y$, after the SFE subroutine.  

Now, to obtain a cheating strategy for Bob in SFE, consider the following. 
Imagine if Bob uses his optimal die rolling strategy to communicate with Alice to create the encoding of $x$ as described above at the end of the SFE protocol. 
Well, we know the average success probability of Bob learning each function value; it is equal to  $\mathrm{B}_{\mathrm{DR,0}}$, as explained above. 
If we now apply the sequential gentle measurement lemma, Lemma~\ref{lem:SGMLformal}, we see that Bob can learn \emph{all} the values of $f(x,y)$ with probability at least  
\begin{equation} 
\mathrm{B}_{\mathrm{DR,0}} - 2 (|Y| - 1) \sqrt{1 - \mathrm{B}_{\mathrm{DR,0}}}. 
\end{equation} 
Since this is a valid strategy for Bob to learn all the values of $f(x,y)$, it is a \emph{lower bound} on $\BSFEP$. 

Collecting all the above pieces of information together, and adding Kitaev's lower bound, we have 
\begin{itemize} 
\item $\ASFE = \mathrm{A}_{\mathrm{DR,0}}$; 
\item $\BSFEP \geq \mathrm{B}_{\mathrm{DR,0}} - 2 (|Y| - 1) \sqrt{1 - \mathrm{B}_{\mathrm{DR,0}}}$; 
\item $\mathrm{A}_{\mathrm{DR,0}} \cdot \mathrm{B}_{\mathrm{DR,0}} \geq \frac{1}{|Y|}$.  
\end{itemize} 
Combining these we get a proof of our main theorem, Theorem~\ref{thm:LowerBound}.  
          
 
\section{Learning and gentle measurements}  \label{lgm section}

In this section we first discuss the gentle measurement lemma and  then generalize the concept to fit our needs. 
Then, we discuss the context in which we consider \emph{learning} and show how to apply our generalized gentle measurement lemma. 

\subsection{Gentle measurements} 
  
Before we dive into gentle measurements, we must first define some essential matrix operations. 
Consider two matrices $A$ and $B \in \mathbb{C}^{m \times n}$. 
The trace inner product is defined as
\begin{equation} \langle A, B\rangle = \text{Tr}\left(A^*B\right)  
\end{equation}
where $A^*$ represents the complex conjugate transpose of $A$.  
The trace norm of a matrix $A$ is given by
\begin{equation}\left\Vert A\right\Vert_{tr} = \text{Tr}(\sqrt{A^*A}).
\end{equation} 
The operator norm of a matrix $A$ is given by 
\begin{equation} 
\| A \|_{op} = \sup \left\{ \| Av \|_2 : \| v \|_2 = 1 \right\} 
\end{equation}  
where $\| v \|_2$ denotes the Euclidean norm $\sqrt{\ip{v}{v}}$. 

The idea behind gentle measurements is that if a measurement operator, when applied to a quantum state, produces a given result with high probability, then the post-measured state will be relatively close to the original state. For our purposes, this allows for more information to be gleaned from the state in a successive measurement. This process is formally scoped below. 

\begin{lemma}[Gentle measurement operator~\cite{W19, W99}] \label{gml}
Consider a density operator $\rho$ and a measurement operator $\Lambda$ where $0 \leq \Lambda \leq I$. Suppose that 
\begin{equation} 
\ip{\Lambda}{\rho} \geq 1 - \varepsilon, 
\end{equation}
where $\varepsilon \in [0,1]$. 
Then we have 
\begin{equation} 
\| \rho - \sqrt{\Lambda} \rho \sqrt{\Lambda} \|_{tr} \leq 2 \sqrt{\varepsilon}.  
\end{equation}  
\end{lemma} 
   
We now use this to prove the following.    
   
\begin{lemma}[Sequential gentle measurement operators] \label{sgml}
Consider a density operator $\rho$ and measurement operators $\Lambda_1, \ldots, \Lambda_n$ where $0 \leq \Lambda_k \leq I$ for each $k \in \{ 1, \ldots, n \}$, where $n \geq 2$. 
Suppose that 
\begin{equation} 
\ip{\Lambda}{\rho} \geq 1 - \varepsilon_k,
\end{equation}
where $\varepsilon_k \in [0,1]$ for each $k \in \{ 1, \ldots, n \}$. 
Then we have 
\begin{equation}
\ip{\rho}{\sqrt{\Lambda_n} \cdots \sqrt{\Lambda_2} \Lambda_1  \sqrt{\Lambda_2} \cdots \sqrt{\Lambda_n}} \geq 1 - \epsilon_1 - 2 \sum_{i=2}^n \sqrt{\varepsilon_i}. 
\end{equation}
\end{lemma} 

\begin{proof} 
We prove this by induction. 
Base case: $n = 2$. 
Consider the following quantity  
\begin{equation}  
| \ip{\rho}{\Lambda_1} - \ip{\rho}{\sqrt{\Lambda_2} \Lambda_1 \sqrt{\Lambda_2}} |  
= 
| \ip{\rho}{\Lambda_1} - \ip{\sqrt{\Lambda_2} \rho \sqrt{\Lambda_2}}{\Lambda_1} | 
= 
| \ip{\rho - \sqrt{\Lambda_2} \rho \sqrt{\Lambda_2}}{\Lambda_1} |. 
\end{equation} 
By applying H\"{o}lder's inequality, we get 
\begin{equation}
| \ip{\rho - \sqrt{\Lambda_2} \rho \sqrt{\Lambda_2}}{\Lambda_1} |
\leq 
\| \rho - \sqrt{\Lambda_2} \rho \sqrt{\Lambda_2} \|_{tr} \| \Lambda_1 \|_{op} 
\leq 
2 \sqrt{\varepsilon_2},  
\end{equation} 
where the last inequality follows from the gentle measurement operator lemma (Lemma~\ref{gml}) and the assumption that $0 \leq \Lambda_1 \leq I$.  
This implies that 
\begin{equation} 
\ip{\rho}{\sqrt{\Lambda_2} \Lambda_1 \sqrt{\Lambda_2}} \geq \ip{\rho}{\Lambda_1} - 2 \sqrt{\varepsilon_2} \geq 1 - \varepsilon_1 - 2 \sqrt{\varepsilon_2}. 
\end{equation} 

Inductive step: Assume it is true up to some $k \in \{ 3, \ldots, n-1 \}$.  
We have, again, that 
\begin{align}  
& | \ip{\rho - \sqrt{\Lambda_{k+1}} \rho \sqrt{\Lambda_{k+1}}}{\sqrt{\Lambda_k} \cdots \sqrt{\Lambda_2} \Lambda_1 \sqrt{\Lambda_2} \cdots \sqrt{\Lambda_k}} | \\ 
& \leq 
\| \rho - \sqrt{\Lambda_{k+1}} \rho \sqrt{\Lambda_{k+1}} \|_{tr} 
\| \sqrt{\Lambda_k} \cdots \sqrt{\Lambda_2} \Lambda_1      \sqrt{\Lambda_2} \cdots \sqrt{\Lambda_k} \|_{op} \\ 
& \leq 
\| \rho - \sqrt{\Lambda_{k+1}} \rho \sqrt{\Lambda_{k+1}} \|_{tr} 
\| \sqrt{\Lambda_k} \|_{op} \cdots \| \sqrt{\Lambda_2} \|_{op} \| \Lambda_1 \|_{op} \| \sqrt{\Lambda_2} \|_{op} \cdots \| \sqrt{\Lambda_k} \|_{op} \\ 
& \leq 
2 \sqrt{\varepsilon_{k+1}},  
\end{align} 
noting that the operator norm is submultiplicative. 
Similar to the base case, this implies that 
\begin{align} 
& \ip{\rho}{\sqrt{\Lambda_{k+1}} \cdots \sqrt{\Lambda_2} \Lambda_1 \sqrt{\Lambda_2} \cdots \sqrt{\Lambda_{k+1}}} \\ 
& \geq 
\ip{\rho}{\sqrt{\Lambda_{k}} \cdots \sqrt{\Lambda_2} \Lambda_1 \sqrt{\Lambda_2} \cdots \sqrt{\Lambda_{k}}} - 2 \sqrt{\varepsilon_{k+1}} \\ 
& \geq \left( 1 - \varepsilon_1 - 2 \sum_{i=2}^k \sqrt{\varepsilon_i} \right) - 2 \sqrt{\varepsilon_{k+1}} \\ 
& = 1 - \varepsilon_1 - 2 \sum_{i=2}^{k+1} \sqrt{\varepsilon_i} 
\end{align} 
as desired. 
\end{proof}
  
Note that this is very similar to, but slightly stronger than, the bound implied by the Quantum Union Bound~\cite{W13, ANSV02, A06}. 
However, since we want constant lower bounds, this version helps us to get better constants.  
    
\subsection{Quantum encodings, and proof of Lemma~\ref{lem:SGMLformal}} 
\label{encodings} 

In this section, we pin down what it means for Bob to learn something about Alice's input. 

We may assume that Alice creates the following state 
\begin{equation} 
\sum_{x \in X} p_x \kb{x}
\end{equation} 
where $p_x$ is the probability of her choosing $x$, then control all her actions on it. 
That is, this is a classical register that Alice holds. 
After some communication, Alice and Bob will share some joint state 
\begin{equation} 
\rho := \sum_{x \in X} p_x \kb{x} \otimes \rho_x
\end{equation}  
where $\rho_x$ is a (quantum) encoding of Alice's bit $x$. 

Suppose Bob wants to learn some information about $x$. 
Well, in a sense, $x$ may not exist in Alice's eyes yet. 
In other words, Bob wants to learn some information about the $x$ Alice ``sees.'' 
We may assume that Alice measures her classical register in the computational basis $\{ N_x : x \in X \}$ to obtain $x$. 

Let us assume that Bob uses the measurement $\{ M_b : b \in B \}$ if he wants to learn the value of the function $f : X \to B$. 
In the context of SFE, this function is of the same form once a $y \in Y$ has been fixed.
Now, we can calculate the probability of Bob successfully learning the function $f$ as 
\begin{equation} 
\left\langle \rho, \sum_{x \in X} N_x \otimes  M_{f(x)} \right\rangle. 
\end{equation} 
Note that the structure of $\rho$ is not really all that important, only so much as to imply that we can assume $N_x$ is a basis measurement. 

Now, suppose that for a function $f_i$, for $i \in \{ 1, \ldots, n \}$, Bob has a POVM 
$\{ M_b^i : b \in B \}$ such that he learns the correct value with probability at least $1 - \varepsilon_i$. 
Then from the above expression, we can write   
\begin{equation} 
\left\langle \rho, \sum_{x \in X} N_x \otimes M^i_{f_i(x)} \right\rangle \geq 1 - \varepsilon_i. 
\end{equation} 
By defining 
\begin{equation} 
\Lambda_i = \sum_{x \in X} N_x \otimes M_{f_i(x)}^i 
\end{equation} 
we can apply Lemma~\ref{sgml} to get that 
\begin{equation} \label{seqlearn}
\ip{\rho}{\sqrt{\Lambda_n} \cdots \sqrt{\Lambda_2} \Lambda_1 \sqrt{\Lambda_2} \cdots \Lambda_n} \geq 1 - \varepsilon_1 - 2 \sum_{i=2}^n \sqrt{\varepsilon_i}.  
\end{equation} 
Now, the neat thing is that since $\{ N_x \}$ is a basis measurement, we have that 
\begin{equation} \label{neato}
\sqrt{\Lambda_n} \cdots \sqrt{\Lambda_2} \Lambda_1 \sqrt{\Lambda_2} \cdots \sqrt{\Lambda_n}
= 
\sum_{x \in X} N_x \otimes \sqrt{M^n_{f_n(x)}} \cdots \sqrt{M^2_{f_2(x)}} 
M^1_{f_1(x)}
\sqrt{M^2_{f_2(x)}} \cdots \sqrt{M^n_{f_n(x)}}. 
\end{equation} 

This suggests we define the POVM 
\begin{equation} 
\{ \tilde{M}_{b_1, \ldots, b_n} : b_1, \ldots, b_n \in B \} 
\end{equation} 
where 
\begin{equation} 
\tilde{M}_{b_1, \ldots, b_n} := 
\sqrt{M^n_{b_n}} \cdots  
\sqrt{M^2_{b_2}} M^1_{b_1} \sqrt{M^2_{b_2}} \cdots \sqrt{M^n_{b_n}}.  
\end{equation} 
One can check that this is a valid POVM and Inequality~\eqref{seqlearn}~and Equation~\eqref{neato} show that this POVM learns $f_i(x)$ for every $i \in \{ 1, \ldots, n \}$, with probability at least 
\begin{equation}
1 - \varepsilon_1 - 2 \sum_{i=2}^n \sqrt{\varepsilon_i}. 
\end{equation} 

Note that since the measurement operators have the POVM $\{ M^1_b : b \in B \}$ ``in the middle,'' and this choice was arbitrary, then we can see that Bob can create another measurement with $\{ M^i_b : b \in B \}$ 
``in the middle'' for any choice of $i$ he wants. 
Thus, if he randomly chooses which measurement is ``in the middle,'' then we see that we can average the success probability as 
\begin{equation}
\frac{1}{n} \sum_{j=1}^n \left( 1 - \varepsilon_j - 2 \sum_{i \neq j}^n \sqrt{\varepsilon_i} \right) = 1 - \frac{\sum_{i=1}^n \varepsilon_i}{n} - \frac{2(n-1)}{n} \sum_{i=1}^n \sqrt{\varepsilon_i}. 
\end{equation} 
Using Cauchy-Schwarz, one can prove that 
\begin{equation}
\sum_{i=1}^n \sqrt{\varepsilon_i} \leq \sqrt{n} \sqrt{\sum_{i=1}^n \varepsilon_i}. 
\end{equation}  
Therefore, the average success probability is bounded below by   
\begin{equation} \label{lastbound} 
1 - \frac{\sum_{i=1}^n \varepsilon_i}{n} - \frac{2(n-1)}{\sqrt n} \sqrt{\sum_{i=1}^n \varepsilon_i}. 
\end{equation} 
In the context of Lemma~\ref{lem:SGMLformal}, we have that $p_i = 1 - \varepsilon_i$ is the guessing probability of learning $f_i(x)$. 
Substituting this into \eqref{lastbound}, we finish our proof of Lemma~\ref{lem:SGMLformal}. 

   
\section{Applications} \label{sect:app}

We now present several applications of our lower bound. 


\subsection{$1$-out-of-$n$ oblivious transfer} \label{OT} 
In $1$-out-of-$n$ oblivious transfer, Alice has an input string $\left(x_1, \ldots, x_n\right) \in W^{\times n}$ (where $|W|$ is finite). Traditionally, we have $W = \{ 0, 1 \}$ so that $x$ is an $n$-bit string, however our bound works in this general setting.  
Bob has a choice input $y \in \{1, \ldots, n\}$. At the end of the protocol, Bob learns $x_y$, i.e., the $y$-th component of $x$. Ideally, Alice should not learn anything about Bob's choice $y$ and Bob should not learn anything more about Alice's input than $x_y$. In the context of SFE, $f(x,y)=x_y$. 
See Figure~\ref{fig_1not2} for a pictorial representation of this application. 

\begin{figure}[h] 
\centering
\begin{tikzpicture}[node distance=1.11cm]
\node[draw,minimum width=8cm, minimum height=2cm] (1NOT2) {$1$-out-of-$n$ OT};
\node (Alice) [above=of 1NOT2.155] {Alice};
\node (Bob) [above=of 1NOT2.22] {Bob};
\draw[->] ([yshift=1cm]1NOT2.155) --  node[fill=white] {\ssmall $\left(x_1, x_2, \ldots, x_{n-1}, x_n\right)\in W^{\times n}$} +(0pt,-1cm);
\draw[->] ([yshift=1cm]1NOT2.22) -- node[fill=white] {\ssmall $y \in \{1,2,\ldots,n-1,n\}$} +(0pt,-1cm);
\draw[->] (1NOT2.338) -- node[fill=white] {\ssmall $x_y$} +(0pt,-1cm);
\end{tikzpicture}
\caption{A pictorial representation of $1$-out-of-$n$ oblivious transfer. Alice has an input $\left(x_1, x_2, \ldots, x_{n-1}, x_n\right)\in W^{\times n}$ and Bob has an input $y \in \{1, 2, \ldots, n-1, n\}$. At the end of the protocol, Bob has learned $x_y$.} 
\label{fig_1not2}
\end{figure}
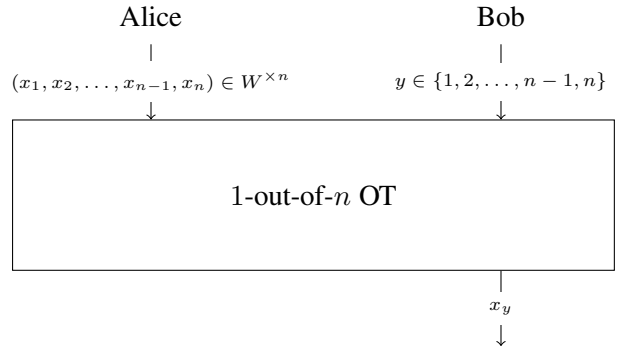

For each such protocol we define the following symbols.
\begin{center}
\begin{tabularx}{\textwidth}{rX}
  $\BOT$: & The maximum probability with which cheating Bob can guess honest Alice's input $x$. \\ 
  $\AOT$: & The maximum probability with which cheating Alice can guess honest Bob's input $y$. \\ 
\end{tabularx}
\end{center}
  
We can apply the lower bound from Inequality~\eqref{lb} if we scope the three values $n = \left|Y\right|$, $\ASFE = \AOT$, and $\BSFEP = \BOT$, to obtain  
\begin{equation}\label{lbOT}
\BOT \geq \frac{1}{n\AOT}-2(n-1)\sqrt{1-\frac{1}{n\AOT}}.\end{equation} 

For this task, we have 
\begin{equation} 
\BSFER = \frac{1}{|W|^{n-1}} 
\end{equation} 
since Bob can learn one $x_y$ and must randomly guess the other $n-1$ values.

Now we revisit the constants $c_A$ and $c_B$. We make the assumptions following Inequality~\eqref{ASFEc} and Equation~\eqref{BSFEc} that
\begin{equation}
\BOT = c_B \cdot \BSFER=\frac{c_B}{\left|W\right|^{n-1}} \quad \text{and} \quad \AOT \leq \frac{c_A}{\left|Y\right|}=\frac{c_A}{n},\end{equation}
which, using Inequality~\eqref{lbOT}, provides us with a lower bound reminiscent of Inequality~\eqref{lb3}, i.e.,
\begin{equation}c_B \geq \left|W\right|^{n-1} \left(\frac{1}{c_A}-2(n-1)\sqrt{1-\frac{1}{c_A}}\right).
\end{equation}
Figure \ref{1nOTc} illustrates the tradeoff between values of $c_B$ and $c_A$ for $\left|W\right|=2$ and varying values of $n$. Additionally, Figure \ref{1nOTc3} illustrates the tradeoff between values of $c_B$ and $c_A$ for $\left|W\right|=3$ and varying values of $n$.

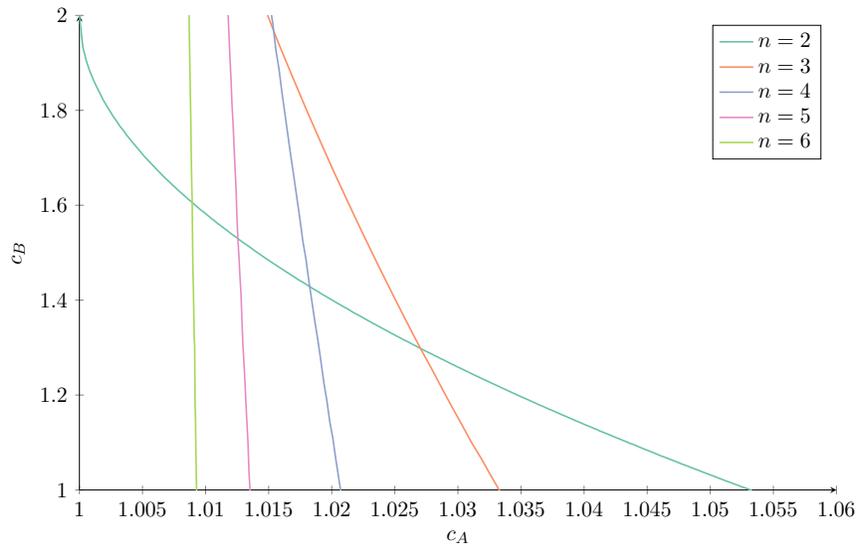
\begin{figure}[hbt!]
\centering
\scalebox{0.75}{%
\begin{tikzpicture}
\begin{axis}[ymin=1, ymax=2, xmin=1, xmax=1.06, axis lines=left, xlabel = {$c_A$}, ylabel={$c_B$}, width=15cm, height=10cm, x tick label style={/pgf/number format/.cd, precision=4, /tikz/.cd}]
\addplot+ [mark=none, thick, domain=1:1.0533, samples=200] {2^(2-1)*(1/(x)-2*(2-1)*sqrt(1-1/(x)))};
\addlegendentry{$n=2$}
\addplot+ [mark=none, thick, domain=1.0148:1.0334] {2^(3-1)*(1/(x)-2*(3-1)*sqrt(1-1/(x)))};
\addlegendentry{$n=3$}
\addplot+ [mark=none, thick, domain=1.0151:1.0208] {2^(4-1)*(1/(x)-2*(4-1)*sqrt(1-1/(x)))};
\addlegendentry{$n=4$}
\addplot+ [mark=none, thick, domain=1.0117:1.0136] {2^(5-1)*(1/(x)-2*(5-1)*sqrt(1-1/(x)))};
\addlegendentry{$n=5$}
\addplot+ [mark=none, thick, domain=1.0086:1.0094] {2^(6-1)*(1/(x)-2*(6-1)*sqrt(1-1/(x)))};
\addlegendentry{$n=6$} 
\end{axis}
\end{tikzpicture}%
}
\caption{$c_A$ vs. $c_B$ for $1$-out-of-$n$ oblivious transfer, $\left|W\right|=2$, and varying $n$.} 
\label{1nOTc}
\end{figure} 

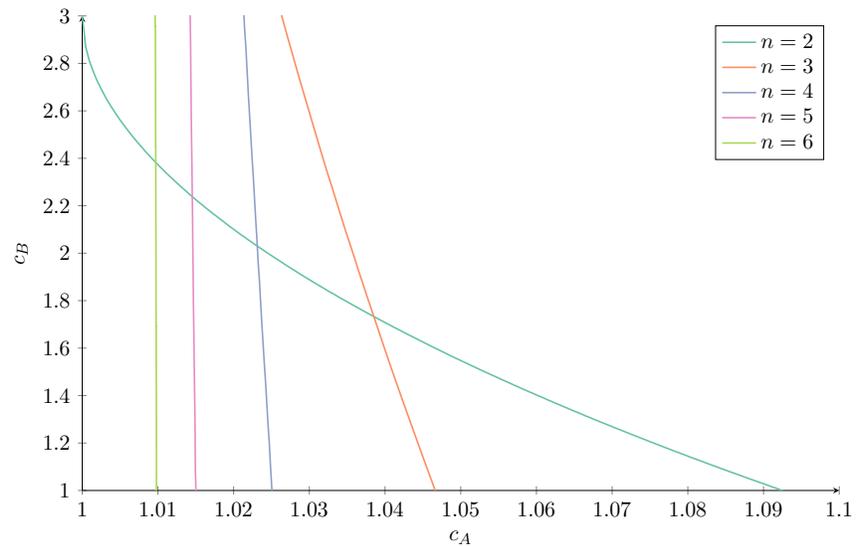
\begin{figure}[hbt!]
\centering
\scalebox{0.75}{%
\begin{tikzpicture}
\begin{axis}[ymin=1, ymax=3, xmin=1, xmax=1.1, axis lines=left, xlabel = {$c_A$}, ylabel={$c_B$}, width=15cm, height=10cm, x tick label style={/pgf/number format/.cd, precision=4, /tikz/.cd}]
\addplot+ [mark=none, thick, domain=1:1.0926, samples=200] {3^(2-1)*(1/(x)-2*(2-1)*sqrt(1-1/(x)))};
\addlegendentry{$n=2$}
\addplot+ [mark=none, thick, domain=1.0263:1.0467] {3^(3-1)*(1/(x)-2*(3-1)*sqrt(1-1/(x)))};
\addlegendentry{$n=3$}
\addplot+ [mark=none, thick, domain=1.0213:1.0252] {3^(4-1)*(1/(x)-2*(4-1)*sqrt(1-1/(x)))};
\addlegendentry{$n=4$}
\addplot+ [mark=none, thick, domain=1.0142:1.0151] {3^(5-1)*(1/(x)-2*(5-1)*sqrt(1-1/(x)))};
\addlegendentry{$n=5$}
\addplot+ [mark=none, thick, domain=1.0096:1.0099] {3^(6-1)*(1/(x)-2*(6-1)*sqrt(1-1/(x)))};
\addlegendentry{$n=6$}
\end{axis}
\end{tikzpicture}%
}
\caption{$c_A$ vs. $c_B$ for $1$-out-of-$n$ oblivious transfer, $\left|W\right|=3$, and varying $n$.}
\label{1nOTc3}
\end{figure}

Using the rationale behind Equation~\eqref{exact}, we can introduce the constant $c$ such that
\begin{equation}\label{1nOTexact}c = \left|W\right|^{n-1} \left(\frac{1}{c}-2(n-1)\sqrt{1-\frac{1}{c}}\right).\end{equation}
 

\subsubsection{Special cases} 
We now look at a few special cases to see how the numbers behave. 

\paragraph{$1$-out-of-$2$ (bit) OT.} 
When $|W| = 2$, $n = 2$, we can calculate $c$ from~\eqref{1nOTexact} to get $c \approx 1.0484$. 
This implies  
\begin{equation} 
\BOT \gtrapprox 0.5242 > 0.5000
\quad \text{ or } \quad 
\AOT \gtrapprox 0.5242 > 0.5000.
\end{equation} 
  

\paragraph{$1$-out-of-$3$ (bit) OT.} 
When $|W| = 2$, $n = 3$, we can calculate $c$ from~\eqref{1nOTexact} to get $c \approx 1.0326$. 
This implies  
\begin{equation} 
\BOT \gtrapprox 0.2581 > 0.2500 
\quad \text{ or } \quad 
\AOT \gtrapprox 0.3442 > 0.3333.
\end{equation} 


\paragraph{$1$-out-of-$2$ (trit) OT.} 
When $|W| = 3$, $n = 2$, we can calculate $c$ from~\eqref{1nOTexact} to get $c \approx 1.085$.  
This implies that 
\begin{equation} 
\BOT \gtrapprox 0.3617 > 0.3333
\quad \text{ or } \quad 
\AOT \gtrapprox 0.5425 > 0.5000.
\end{equation} 

\paragraph{A bit of history.} 

In terms of cheating probabilities, there are a few bounds which have been concluded for oblivious transfer. 
In~\cite{CKS13}, it was shown for the case of $|W| = 2$, $n=2$ that $\max\{ \BOT, \AOT \} \gtrapprox 0.5852$ using a lower bound on bit commitment. 
The lower bound constant was later improved to $2/3$ in~\cite{CGS16}. 
In~\cite{GRS18}, it was shown that for $|W| = 2^m$, for any $m$, and $n=2$, we have $\max\{ \BOT, \AOT \} \gtrapprox 0.61$ (which is independent of $m$).  
The numbers we presented above do not improve upon the known lower bound for $|W| = 2$ and we suspect will not improve upon the known lower bound for $|W| = 2^m$ for any $m$. 
However, the bounds in those papers are specific to both oblivious transfer and to the case of $n=2$, whereas our bound is much more general. 
As far as we are aware, the case of $n > 2$ has not been explored (at least in this context) so any bound (such as the one above) is new.   
 

\subsection{$k$-out-of-$n$ oblivious transfer} \label{KNOT}
   
In $k$-out-of-$n$ oblivious transfer, Alice has an input string $\left(x_1, \ldots, x_n\right) \in W^{\times n}$. 
Bob has a choice input $y \in Y = \{ S : S \subset \{ 1, \ldots, n \}, |S| = k \}$, for $k < n$. 
At the end of the protocol, Bob learns the components of $x$ corresponding to the set $y$, i.e., $\{ x_i : i \in y \}$. 
Ideally, Alice should not learn anything about Bob's choice of proper subset $y$ and Bob should not learn anything more about Alice's input than $\{x_i : i \in y\}$. 
In the context of SFE, $f(x,y)=\{ x_i : i \in y \}$. See Figure~\ref{fig_knot2} for a pictorial representation of this application. 

\begin{figure}[h]
\centering
\begin{tikzpicture}[node distance=1.11cm]
\node[draw,minimum width=8cm, minimum height=2cm] (KNOT2) {${n \choose k}$ OT};
\node (Alice) [above=of KNOT2.155] {Alice};
\node (Bob) [above=of KNOT2.22] {Bob};
\draw[->] ([yshift=1cm]KNOT2.155) --  node[fill=white] {\ssmall $\left(x_1, x_2, \ldots, x_{n-1}, x_n\right)\in W^{\times n}$} +(0pt,-1cm);
\draw[->] ([yshift=1cm]KNOT2.22) -- node[fill=white] {\ssmall $y \subset \{1,2,\ldots,n-1,n\}$} +(0pt,-1cm);
\draw[->] (KNOT2.338) -- node[fill=white] {\ssmall $\{x_i : i \in y\}$} +(0pt,-1cm);
\end{tikzpicture}
\caption{A pictorial representation of $k$-out-of-$n$ oblivious transfer. 
Alice has an input $\left(x_1, x_2, \ldots, x_{n-1}, x_n\right)\in W^{\times n}$ and Bob has an input $y \subset \{1, 2, \ldots, n-1, n\}$ of size $k$. 
At the end of the protocol, Bob has learned $\{x_i : i \in y\}$.} 
\label{fig_knot2}
\end{figure}
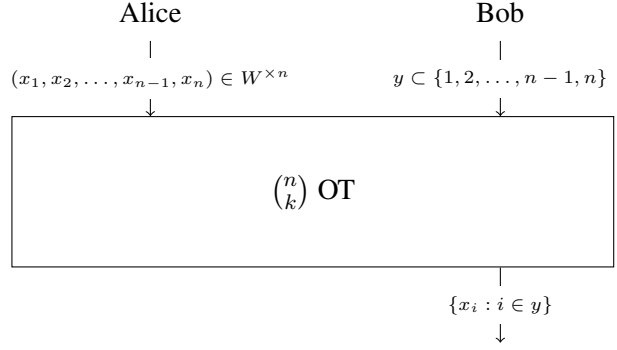
  
For each such protocol we define the following symbols.
\begin{center}
\begin{tabularx}{\textwidth}{rX}
  $\BKNOT$: & The maximum probability with which cheating Bob can guess honest Alice's input $x$. \\
  $\AKNOT$: & The maximum probability with which cheating Alice can guess honest Bob's input $y$. \\ 
\end{tabularx}
\end{center}
 
We can apply the lower bound from Inequality~\eqref{lb} noting that $\ASFE = \AKNOT$ and $\BSFEP = \BKNOT$\footnote{It is worth noting here that Bob seems to be \emph{over learning} the value of $x$ in SFE.  
For instance, when $k = 2$ and $n = 3$, $\BSFEP$ wants him to learn \emph{every} subset pair, i.e., $\{ (x_1, x_2), (x_2, x_3), (x_1, x_3) \}$, from which $x$ can be inferred. 
On the other hand, learning $x$ implies the knowledge about every subset, so the probabilities are the same.}.  
We can now apply the lower bound, and we get the following inequality 
\begin{equation} \label{lbKNOT}
\BKNOT \geq \frac{1}{{n \choose k} \AKNOT}-2\left({n \choose k}-1\right)\sqrt{1-\frac{1}{{n \choose k} \AKNOT}}. 
\end{equation}  
  
For this task, we have 
\begin{equation} 
\BSFER = \frac{1}{|W|^{n-k}} 
\end{equation} 
since Bob can learn $k$ values of Alice's input and must randomly guess the other $n-k$ values.

Now we revisit the constants $c_A$ and $c_B$. We make the assumptions following Inequality~\eqref{ASFEc} and Equation~\eqref{BSFEc} that
\begin{equation}
\BKNOT = c_B \cdot \BSFER=\frac{c_B}{\left|W\right|^{n-k}} \quad \text{and} \quad \AKNOT \leq \frac{c_A}{\left|Y\right|}=\frac{c_A}{{n \choose k}},\end{equation}
which, using Inequality~\eqref{lbKNOT}, provides us with a lower bound reminiscent of Inequality~\eqref{lb3}, i.e.,
\begin{equation}c_B \geq \left|W\right|^{n-k} \left(\frac{1}{c_A}-2\left({n \choose k}-1\right)\sqrt{1-\frac{1}{c_A}}\right).
\end{equation}
Figure \ref{knOTc} illustrates the tradeoff between values of $c_B$ and $c_A$ for $\left|W\right|=2$, $k=2$, and varying values of $n$. Additionally, Figure \ref{knOTcnhalf} illustrates the tradeoff between $c_B$ and $c_A$ for $\left|W\right|=2$, $k=\frac{n}{2}$, and varying values of $n$.

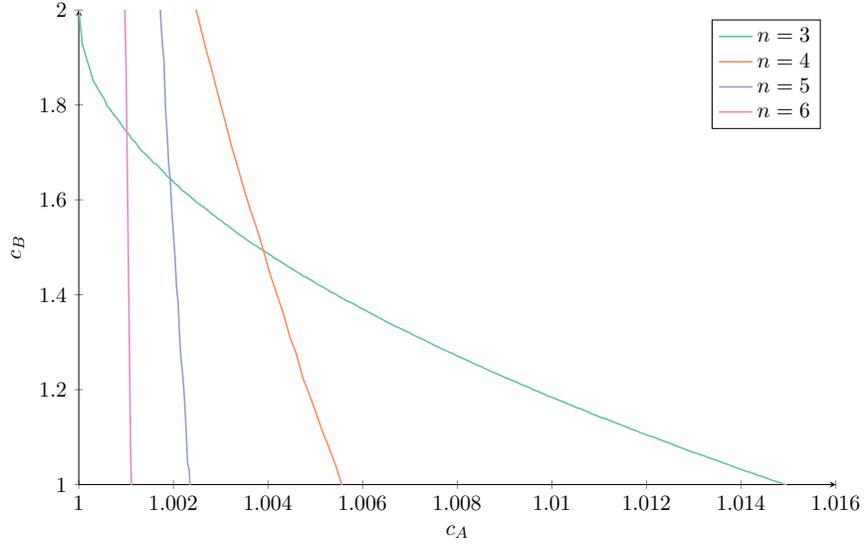
\begin{figure}[hbt!]
\centering
\scalebox{0.75}{%
\begin{tikzpicture}
\begin{axis}[ymin=1, ymax=2, xmin=1, xmax=1.016, axis lines=left, xlabel = {$c_A$}, ylabel={$c_B$}, width=15cm, height=10cm, x tick label style={/pgf/number format/.cd, precision=4, /tikz/.cd}]
\addplot+ [mark=none, thick, domain=1:1.015, samples=200] {2^(3-2)*(1/(x)-2*(3-1)*sqrt(1-1/(x)))};
\addlegendentry{$n=3$}
\addplot+ [mark=none, thick, domain=1.0024:1.0057] {2^(4-2)*(1/(x)-2*(6-1)*sqrt(1-1/(x)))};
\addlegendentry{$n=4$}
\addplot+ [mark=none, thick, domain=1.0017:1.0025] {2^(5-2)*(1/(x)-2*(10-1)*sqrt(1-1/(x)))};
\addlegendentry{$n=5$}
\addplot+ [mark=none, thick, domain=1:1.0012] {2^(6-2)*(1/(x)-2*(15-1)*sqrt(1-1/(x)))};
\addlegendentry{$n=6$} 
\end{axis}
\end{tikzpicture}%
}
\caption{$c_A$ vs. $c_B$ for $k$-out-of-$n$ oblivious transfer, $\left|W\right|=2$, $k=2$, and varying $n$.}
\label{knOTc}
\end{figure}

\begin{figure}[hbt!]
\centering
\scalebox{0.75}{%
\begin{tikzpicture}
\begin{axis}[ymin=1, ymax=2, xmin=1, xmax=1.06, axis lines=left, xlabel = {$c_A$}, ylabel={$c_B$}, width=15cm, height=10cm, x tick label style={/pgf/number format/.cd, precision=4, /tikz/.cd}]
\addplot+ [mark=none, thick, domain=1:1.0532, samples=200] {2^(2-1)*(1/(x)-2*(2-1)*sqrt(1-1/(x)))};
\addlegendentry{$n=2$, $k=1$}
\addplot+ [mark=none, thick, domain=1.0024:1.0056] {2^(4-2)*(1/(x)-2*(6-1)*sqrt(1-1/(x)))};
\addlegendentry{$n=4$, $k=2$}
\addplot+ [mark=none, thick, domain=1.0003:1.0006] {2^(6-3)*(1/(x)-2*(20-1)*sqrt(1-1/(x)))};
\addlegendentry{$n=6$, $k=3$} 
\end{axis}
\end{tikzpicture}%
}
\caption{$c_A$ vs. $c_B$ for $k$-out-of-$n$ oblivious transfer, $\left|W\right|=2$, $k=\frac{n}{2}$, and varying $n$.}
\label{knOTcnhalf}
\end{figure}
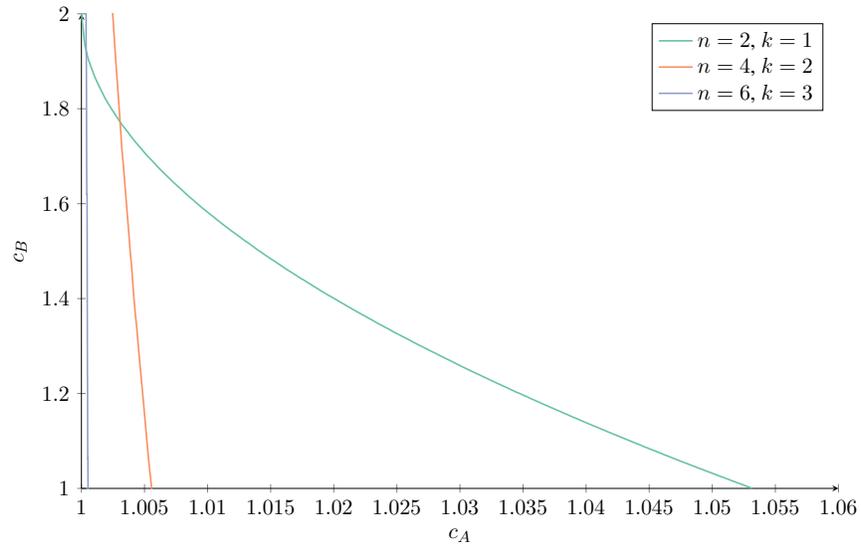

Using the rationale behind Equation~\eqref{exact}, we can introduce the constant $c$ such that
\begin{equation}\label{knOTexact}c = \left|W\right|^{n-k} \left(\frac{1}{c}-2\left({n\choose k}-1\right)\sqrt{1-\frac{1}{c}}\right).\end{equation}
  

\newpage 
\subsubsection{Special cases} 
We now look at two special cases to see how the numbers behave. 

 
\paragraph{$2$-out-of-$3$ (bit) OT.} 
When $|W| = 2$, $n = 3$, and $k=2$, we can calculate $c$ from~\eqref{knOTexact} to get $c \approx 1.0145$.  
This implies that 
\begin{equation} 
\BKNOT \gtrapprox 0.5073 > 0.5000 
\quad \text{ or } \quad 
\AKNOT \gtrapprox 0.3382 > 0.3333.
\end{equation} 

 
\paragraph{$2$-out-of-$4$ (bit) OT.} 
When $|W| = 2$, $n = 4$, and $k = 2$, we can calculate $c$ from~\eqref{knOTexact} to get $c \approx 1.0056$.  
This implies that 
\begin{equation} 
\BKNOT \gtrapprox 0.2514 > 0.2500 
\quad \text{ or } \quad 
\AKNOT \gtrapprox 0.1676 > 0.1667.
\end{equation} 


\paragraph{$3$-out-of-$4$ (bit) OT.} 
When $|W| = 2$, $n = 4$, and $k=3$, we can calculate $c$ from~\eqref{knOTexact} to get $c \approx 1.0067$.  
This implies that 
\begin{equation} 
\BKNOT \gtrapprox 0.5034 > 0.5000
\quad \text{ or } \quad 
\AKNOT \gtrapprox 0.2517 > 0.2500.
\end{equation} 


\paragraph{A bit of history.} 

As far as we are aware, this primitive has not been studied before in the quantum literature in this context so the bounds we present here on the cheating probabilities are the first of their kind. 
We also note that Bob's cheating probability we give here is likely very much improvable if one were to specialize our proof to this specific class of protocols. 
However, we leave this open to future work. 


\subsection{XOR oblivious transfer} \label{XOT} 
In XOR oblivious transfer, Alice has two $n$-bit string inputs $x_1 \in \{0,1\}^n$ and $x_2 \in \{0,1\}^n$. 
Bob has a choice input $y \in \{1, 2, \oplus \}$. 
At the end of the protocol, Bob learns $x_1$ if $y=1$, $x_2$ if $y=2$, or $x_1 \oplus x_2$ if $y=\oplus$. Ideally, Alice should not learn anything about Bob's choice $y$ and Bob should not learn anything more about Alice's inputs than his output. 
In the context of SFE, 

\[f(x,y)=
\begin{cases}
x_y & \text{if }y\in\{1,2\},\\
x_1\oplus x_2 & \text{if }y = \oplus. 
\end{cases}\]
See Figure~\ref{fig_xot} for a pictorial representation of this application. 

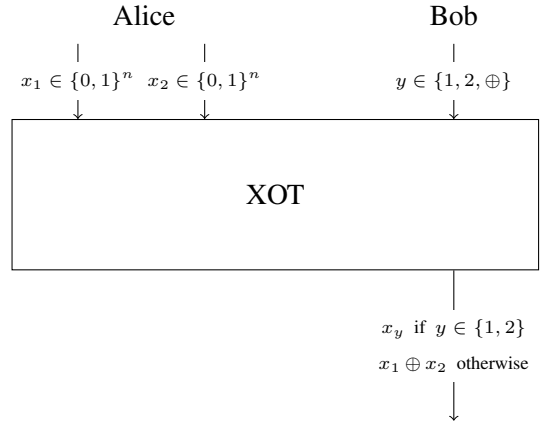
\begin{figure}[h]
\centering
\begin{tikzpicture}[node distance=1.11cm]
\node[draw,minimum width=7cm, minimum height=2cm] (XOT) {XOT};
\node (Alice) [above=of XOT.150] {Alice};
\node (Bob) [above=of XOT.23] {Bob};
\draw[->] ([yshift=1cm]XOT.159) --  node[fill=white] {\ssmall $x_1 \in \{0,1\}^n$} +(0pt,-1cm);
\draw[->] ([yshift=1cm]XOT.133) --  node[fill=white] {\ssmall $x_2 \in \{0,1\}^n$} +(0pt,-1cm);
\draw[->] ([yshift=1cm]XOT.23) -- node[fill=white] {\ssmall $y \in \{1,2, \oplus\}$} +(0pt,-1cm);
\draw[->] (XOT.337) -- node[fill=white, text width =2cm, align=center] {\ssmall $x_y$ if $y\in \{1,2\}$ $x_1 \oplus x_2 \text{ otherwise}$} +(0pt,-2cm);
\end{tikzpicture}
\caption{A pictorial representation of XOR oblivious transfer. Alice has two $n$-bit input strings $x_1, x_2 \in \{0, 1\}^n$ and Bob has a choice value $y \in \{1, 2, \oplus\}$. At the end of the protocol, Bob has learned $x_y$ if $y\in \{1,2\}$, or $x_1 \oplus x_2$ otherwise.} 
\label{fig_xot}
\end{figure}

For each such protocol we define the following symbols.
\begin{center}
\begin{tabularx}{\textwidth}{rX}
  $\BXOT$: & The maximum probability with which cheating Bob can guess honest Alice's inputs $x_1$ and $x_2$. \\
  $\AXOT$: & The maximum probability with which cheating Alice can guess honest Bob's input $y$. \\ 
\end{tabularx}
\end{center} 

We can apply the lower bound from Inequality~\eqref{lb} if we scope the three values $\left|Y\right| = 3$, $\ASFE = \AXOT$, and $\BSFEP = \BOT$. 
We can now apply the lower bound, and we get the following inequality 
\begin{equation}\label{lbXOT}
\BXOT \geq \frac{1}{3 \AXOT} - 4\sqrt{1-\frac{1}{3 \AXOT}}.  
\end{equation} 
   
For this task, we have 
\begin{equation} 
\BSFER = \frac{1}{2^n} 
\end{equation} 
since Bob can learn one of Alice's $n$-bit inputs entirely, and must randomly guess the other.

Now we revisit the constants $c_A$ and $c_B$. We make the assumptions following Inequality~\eqref{ASFEc} and Equation~\eqref{BSFEc} that
\begin{equation}
\BXOT = c_B \cdot \BSFER=\frac{c_B}{2^n} \quad \text{and} \quad \AXOT \leq \frac{c_A}{\left|Y\right|}=\frac{c_A}{3},\end{equation}
which, using Inequality~\eqref{lbXOT}, provides us with a lower bound reminiscent of Inequality~\eqref{lb3}, i.e.,
\begin{equation}c_B \geq 2^n \left(\frac{1}{c_A}-4\sqrt{1-\frac{1}{c_A}}\right).
\end{equation}
Figure \ref{XOTc} illustrates the tradeoff between values of $c_B$ and $c_A$ for varying values of $n$.

\begin{figure}[hbt!] 
\centering
\scalebox{0.75}{%
\begin{tikzpicture}
\begin{axis}[ymin=1, ymax=2, xmin=1, xmax=1.07, axis lines=left, xlabel = {$c_A$}, ylabel={$c_B$}, width=15cm, height=10cm, x tick label style={/pgf/number format/.cd, precision=4, /tikz/.cd}]
\addplot+ [mark=none, thick, domain=1:1.015, samples=200] {2^(1)*(1/(x)-4*sqrt(1-1/(x)))};
\addlegendentry{$n=1$}
\addplot+ [mark=none, thick, domain=1.0149:1.0334] {2^(2)*(1/(x)-4*sqrt(1-1/(x)))};
\addlegendentry{$n=2$}
\addplot+ [mark=none, thick, domain=1.0333:1.0453] {2^(3)*(1/(x)-4*sqrt(1-1/(x)))};
\addlegendentry{$n=3$}
\addplot+ [mark=none, thick, domain=1.0452:1.052] {2^(4)*(1/(x)-4*sqrt(1-1/(x)))};
\addlegendentry{$n=4$}
\addplot+ [mark=none, thick, domain=1.0518:1.0555] {2^(5)*(1/(x)-4*sqrt(1-1/(x)))};
\addlegendentry{$n=5$} 
\end{axis}
\end{tikzpicture}%
}
\caption{$c_A$ vs. $c_B$ for XOR oblivious transfer and varying $n$.}
\label{XOTc}
\end{figure}
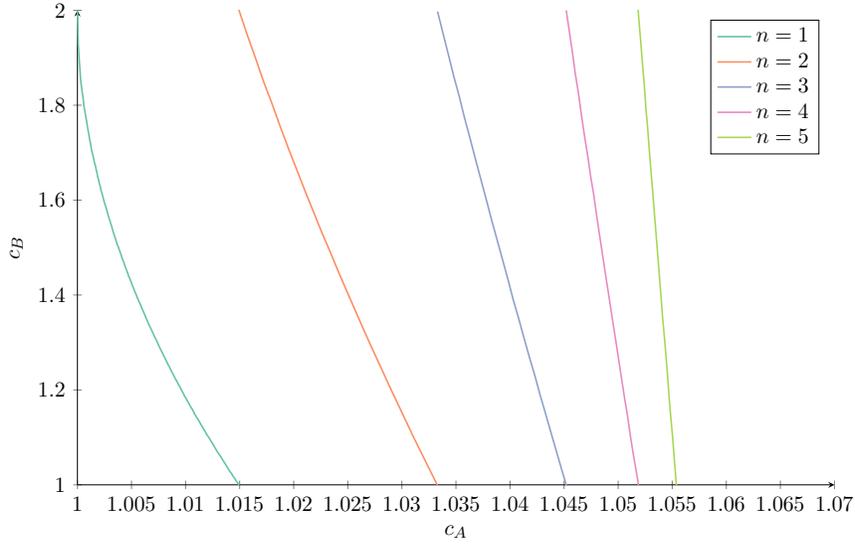

Using the rationale behind Equation~\eqref{exact}, we can introduce the constant $c$ such that
\begin{equation}\label{XOTexact}c = 2^n \left(\frac{1}{c}-4\sqrt{1-\frac{1}{c}}\right).\end{equation}
   

\subsubsection{Special cases} 
We now look at some special cases to see how the numbers behave. 

 
\paragraph{1-bit XOR oblivious transfer.} 
When $n=1$, we can calculate $c$ from~\eqref{XOTexact} to get $c \approx 1.0145$.  
This implies that 
\begin{equation} 
\BXOT \gtrapprox 0.5073 > 0.5000
\quad \text{ or } \quad 
\AXOT \gtrapprox 0.3382 > 0.3333.
\end{equation}

\paragraph{\indent Remark:} Coincidentally, this is the same trade-off as in the case of $2$-out-of-$3$ bit oblivious transfer. 

 
\paragraph{2-bit XOR oblivious transfer.} 
When $n=2$, we can calculate $c$ from~\eqref{XOTexact} to get $c \approx 1.0326$.  
This implies that 
\begin{equation} 
\BXOT \gtrapprox 0.2582 > 0.2500
\quad \text{ or } \quad 
\AXOT \gtrapprox 0.3442 > 0.3333.
\end{equation}


\paragraph{A bit of history.} 

XOR oblivious transfer is less studied than its $1$-out-of-$n$ sibling. 
The only reference in the quantum literature of which we are aware is \cite{KST20} 
where they present a device-independent protocol achieving $\max \{\BXOT, \AXOT \} < 1$ when $n = 2$. 
The bounds above exhibit the first lower bounds on this task.  
   

\subsection{Equality/{one-way oblivious identification}} \label{EQ}

The equality task is a special case of SFE where $f$ is simply the equality function $f(x,y) = \delta_{x,y}$, i.e., the Kronecker delta function. 
More specifically, Alice has an input $x \in X$, and Bob has an input $y \in Y$, where $X = Y$ here. 
Ideally, Bob should not learn anything more about Alice's input than what comes naturally from the output of the equality function.
The output, denoted $\delta_{xy}$, is a value that indicates to Bob whether their inputs are equal.
See Figure~\ref{fig_eq} for a pictorial representation of this application. 

\begin{figure}[h]
\centering
\begin{tikzpicture}[node distance=1.11cm]
\node[draw,minimum width=4cm, minimum height=2cm] (EQ) {EQ};
\node (Alice) [above=of EQ.145] {Alice};
\node (Bob) [above=of EQ.31] {Bob};
\draw[->] ([yshift=1cm]EQ.147) --  node[fill=white] {\ssmall $x \in X$} +(0pt,-1cm);
\draw[->] ([yshift=1cm]EQ.33) -- node[fill=white] {\ssmall $y \in Y$} +(0pt,-1cm); 
\draw[->] (EQ.328) -- node[fill=white] {\ssmall $\delta_{xy}$} +(0pt,-1cm);
\end{tikzpicture}
\caption{A pictorial representation of the equality function. Alice has an input $x \in X$ and Bob has an input $y \in Y$. 
At the end of the protocol, Bob outputs $\delta_{xy}$, the value of which tells Bob whether their inputs are equal.} 
\label{fig_eq}
\end{figure}
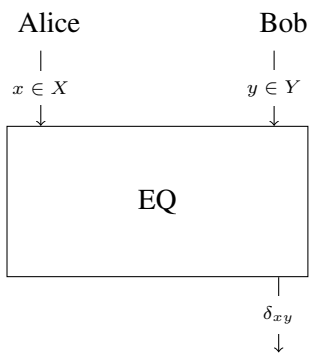 

This task is sometimes referred to as \emph{one-way oblivious identification}. 
You could imagine Bob wants to verify the identity of Alice by seeing if she knows a password.
If the passwords match, then Bob is assured that the other party is Alice. 
An imposter (i.e. a cheating party) would like to use such a protocol to learn Alice's password. 

For each such protocol we define the following symbols.
\begin{center}
\begin{tabularx}{\textwidth}{rX}
  $\BEQ$: & The maximum probability with which cheating Bob can guess honest Alice's input $x$. \\
  $\AEQ$: & The maximum probability with which cheating Alice can guess honest Bob's input $y$. \\ 
\end{tabularx}
\end{center}
 
We can apply the lower bound from Inequality~\eqref{lb} noting that $\ASFE = \AEQ$ and $\BSFEP = \BEQ$. 
Suppose $n = |X| = |Y|$ for clarity. 
We can now apply the lower bound, and we get the following inequality   
\begin{equation} \label{lbEQ}
\BEQ \geq \frac{1}{n \AEQ}-2\left(n-1\right)\sqrt{1-\frac{1}{n \AEQ}}.  
\end{equation}  

For this task, we have 
\begin{equation} 
\BSFER = \frac{2}{n}.
\end{equation} 
Bob's output function $\delta_{xy}$ will return 1 if Alice and Bob's inputs match, with probability $\frac{1}{n}$. Bob's output function $\delta_{xy}$ will return 0 if Alice and Bob's inputs don't match, with probability $\frac{n-1}{n}$, and Bob's random guessing probability is then $\frac{1}{n-1}$, since he already knows his input to the protocol does not match Alice's. Combining the two probabilities gives Bob's random guessing probability of $\frac{1}{n}+\frac{n-1}{n}\cdot\frac{1}{n-1}$ = $\frac{2}{n}$.

Now we revisit the constants $c_A$ and $c_B$. We make the assumptions following Inequality~\eqref{ASFEc} and Equation~\eqref{BSFEc} that
\begin{equation}
\BEQ = c_B \cdot \BSFER=\frac{2c_B}{n} \quad \text{and} \quad \AEQ \leq \frac{c_A}{\left|Y\right|}=\frac{c_A}{n},\end{equation}
which, using Inequality~\eqref{lbEQ}, provides us with a lower bound reminiscent of Inequality~\eqref{lb3}, i.e.,
\begin{equation}c_B \geq \frac{n}{2} \left(\frac{1}{c_A}-2\left(n-1\right)\sqrt{1-\frac{1}{c_A}}\right).
\end{equation}

Note that when $n=2$, $\BEQ$ is necessarily $1$. If $\delta_{xy}$ returns 0, meaning Bob and Alice's inputs are not equal, Bob still knows which input Alice has, since there is only one other element in $X$.

Figure \ref{EQc} illustrates the tradeoff between values of $c_B$ and $c_A$ for varying values of $n$.

\begin{figure}[h]
\centering
\scalebox{0.75}{%
\begin{tikzpicture}
\begin{axis}[ymin=1, ymax=2, xmin=1, xmax=1.008, axis lines=left, xlabel = {$c_A$}, ylabel={$c_B$}, width=15cm, height=10cm, x tick label style={/pgf/number format/.cd, precision=4, /tikz/.cd}] 
\addplot+ [mark=none, thick, domain=1:1.059, samples=200] {(3*(1/(x)-2*(3-1)*sqrt(1-1/(x))))/2};
\addlegendentry{$n=3$}
\addplot+ [mark=none, thick, domain=1:1.027, samples=100] {(4*(1/(x)-2*(4-1)*sqrt(1-1/(x))))/2};
\addlegendentry{$n=4$}
\addplot+ [mark=none, thick, domain=1:1.0154, samples=100] {(5*(1/(x)-2*(5-1)*sqrt(1-1/(x))))/2};
\addlegendentry{$n=5$}
\addplot+ [mark=none, thick, domain=1:1.0099, samples=100] {(6*(1/(x)-2*(6-1)*sqrt(1-1/(x))))/2};
\addlegendentry{$n=6$} 
\end{axis}
\end{tikzpicture}%
}
\caption{$c_A$ vs. $c_B$ for equality and varying $n$.}
\label{EQc}
\end{figure}
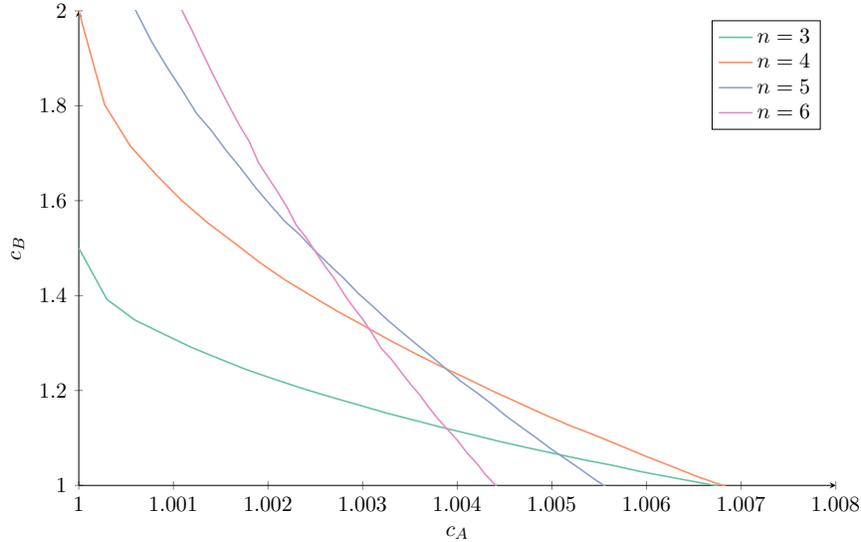

Using the rationale behind Equation~\eqref{exact}, we can introduce the constant $c$ such that
\begin{equation}\label{EQexact}c = \frac{n}{2} \left(\frac{1}{c}-2\left(n-1\right)\sqrt{1-\frac{1}{c}}\right).\end{equation}
   

\subsubsection{Special case} 
We now look at a special case to see how the numbers behave. 

 
\paragraph{Equality function for $n=3$.} 
When $n=3$, we can calculate $c$ from~\eqref{EQexact} to get $c \approx 1.0065$.  
This implies that 
\begin{equation} 
\BEQ \gtrapprox 0.671 > 0.667
\quad \text{ or } \quad 
\AEQ \gtrapprox 0.3355 > 0.3333.
\end{equation}


These are the only lower bounds on the cheating probabilities for equality of which we are aware. 


\subsection{Inner product} \label{IP}

The inner product function is a special case of SFE where $X = \{ 0, 1 \}^n$, $Y = \{ 0, 1 \}^n \setminus \{ 0 \}^n$, and $f$ is simply the inner product function 
\begin{equation}
f(x,y) = x \cdot y := \sum_{i=1}^n x_i y_i \mod 2. 
\end{equation}  
Ideally, Bob should not learn anything more about Alice's inputs than what comes naturally from the output of the inner product function. 
Note that when $n = 2$ here, this is the same as XOR oblivious transfer when the length of each of Alice's input bit strings is $1$. 
They are only the same in this smallest case though. See Figure~\ref{fig_ip} for a pictorial representation of this application. 

\begin{figure}[h]
\centering
\begin{tikzpicture}[node distance=1.11cm]
\node[draw,minimum width=6cm, minimum height=2cm] (IP) {IP};
\node (Alice) [above=of IP.156] {Alice};
\node (Bob) [above=of IP.33] {Bob};
\draw[->] ([yshift=1cm]IP.156) --  node[fill=white] {\ssmall $x \in \{0,1\}^n$} +(0pt,-1cm);
\draw[->] ([yshift=1cm]IP.33) -- node[fill=white] {\ssmall $y \in \{0,1\}^n \setminus \{0\}^n$} +(0pt,-1cm);
\draw[->] (IP.327) -- node[fill=white] {\ssmall $x \cdot y = \sum_{i=1}^n x_i y_i\text{ mod } 2$} +(0pt,-1cm);
\end{tikzpicture}
\caption{A pictorial representation of inner product. Alice has an $n$-bit string $x \in \{0, 1\}^n$ and Bob has an $n$-bit string choice $y \in \{0, 1\}^n \setminus \{0\}^n$. At the end of the protocol, Bob has learned $x \cdot y = \sum_{i=1}^n x_i y_i\text{ mod } 2$.} 
\label{fig_ip}
\end{figure}
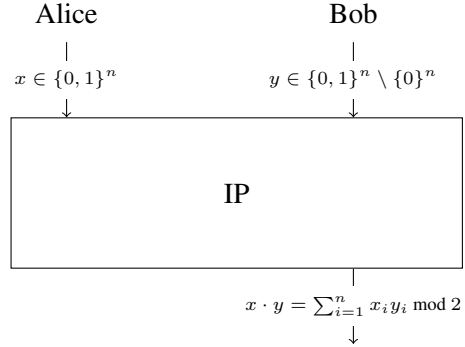

For each such protocol we define the following symbols.
\begin{center}
\begin{tabularx}{\textwidth}{rX}
  $\BIP$: & The maximum probability with which cheating Bob can guess honest Alice's input $x$. \\
  $\AIP$: & The maximum probability with which cheating Alice can guess honest Bob's input $y$. \\ 
\end{tabularx}
\end{center}
 
We can apply the lower bound from Inequality~\eqref{lb} noting that $\ASFE = \AIP$ and $\BSFEP = \BIP$. 
We can now apply the lower bound, and we get the following inequality    
\begin{equation} \label{lbIP}
\BIP \geq \frac{1}{(2^n -1) \AIP} - 2 (2^n -2) \sqrt{1-\frac{1}{(2^n -1) \AIP}}.  
\end{equation} 

For this task, we have 
\begin{equation} 
\BSFER = \frac{2}{2^n}.
\end{equation} 
This is because $|\{ x \in \{ 0, 1 \}^n : x \cdot y = c \}| = 2^{n-1}$ for any $c \in \{ 0, 1 \}$ and nonzero $y$. 

Now we revisit the constants $c_A$ and $c_B$. We make the assumptions following Inequality~\eqref{ASFEc} and Equation~\eqref{BSFEc} that
\begin{equation}
\BIP = c_B \cdot \BSFER=\frac{2c_B}{2^n} \quad \text{and} \quad \AIP \leq \frac{c_A}{\left|Y\right|}=\frac{c_A}{2^n - 1},\end{equation}
which, using Inequality~\eqref{lbEQ}, provides us with a lower bound reminiscent of Inequality~\eqref{lb3}, i.e.,
\begin{equation}c_B \geq \frac{2^n}{2} \left(\frac{1}{c_A}-2\left(2^n-2\right)\sqrt{1-\frac{1}{c_A}}\right).
\end{equation}
Figure \ref{IPc} illustrates the tradeoff between values of $c_B$ and $c_A$ for varying values of $n$. 

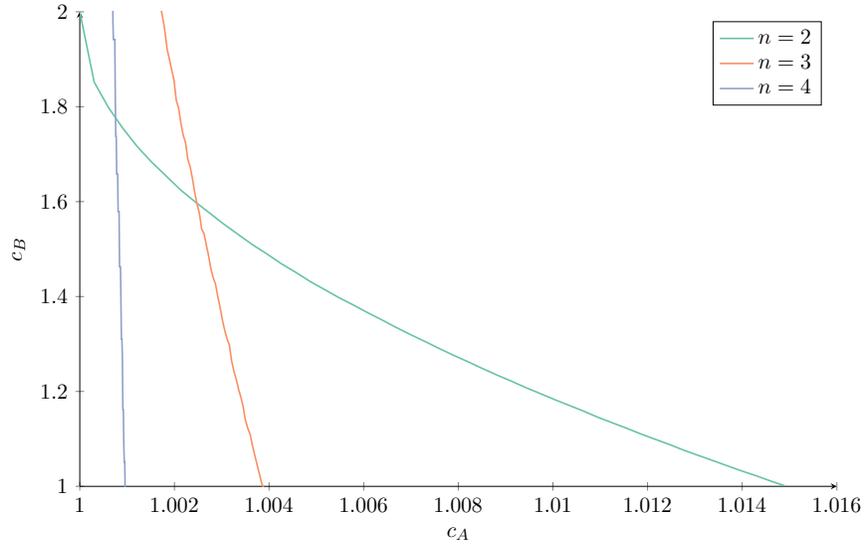
\begin{figure}[h]
\centering
\scalebox{0.75}{%
\begin{tikzpicture}
\begin{axis}[ymin=1, ymax=2, xmin=1, xmax=1.016, axis lines=left, xlabel = {$c_A$}, ylabel={$c_B$}, width=15cm, height=10cm, x tick label style={/pgf/number format/.cd, precision=4, /tikz/.cd}]
\addplot+ [mark=none, thick, domain=1:1.0149, samples=50] {((2^(2))/2)*(1/(x)-2*(2^2-2)*sqrt(1-1/(x)))};
\addlegendentry{$n=2$}
\addplot+ [mark=none, thick, domain=1.0016:1.004, samples=50] {((2^(3))/2)*(1/(x)-2*(2^3-2)*sqrt(1-1/(x)))};
\addlegendentry{$n=3$}
\addplot+ [mark=none, thick, domain=1.0006:1.0011, samples=50] {((2^(4))/2)*(1/(x)-2*(2^4-2)*sqrt(1-1/(x)))};
\addlegendentry{$n=4$}
\end{axis}
\end{tikzpicture}%
}
\caption{$c_A$ vs. $c_B$ for inner product and varying $n$.}
\label{IPc}
\end{figure}

Using the rationale behind Equation~\eqref{exact}, we can introduce the constant $c$ such that
\begin{equation}\label{IPexact}c = \frac{2^n}{2} \left(\frac{1}{c}-2\left(2^n-2\right)\sqrt{1-\frac{1}{c}}\right).\end{equation}


\subsubsection{Special case} 
We now look at a special case to see how the numbers behave. 

 
\paragraph{Inner product function for $n=3$.} 
When $n=3$, we can calculate $c$ from~\eqref{IPexact} to get $c \approx 1.0039$.  
This implies that 
\begin{equation} 
\BIP \gtrapprox 0.251 > 0.250
\quad \text{ or } \quad 
\AIP \gtrapprox 0.1434 > 0.1429.
\end{equation}


These are the only lower bounds on the cheating probabilities for inner product of which we are aware. 

\subsection{Millionaire's problem} \label{MP}

The millionaire's problem is a special case of SFE where $X = W$ and $Y = W \setminus \{n\}$ where $W = \{ 1, \ldots, n \}$ with $n$ being possibly very large, and 
\begin{equation}
f(x,y) = \left\{ \begin{array}{rl} 
1 & \text{ if } y \geq x \\ 
0 & \text{ otherwise}. 
\end{array} \right. 
\end{equation}
We restrict Bob's input since $f(x,n) = 1$ for all $x$, and so is meaningless. Also, it makes sense that honest Bob would not play if he has the maximum amount of money anyway. On the other hand, dishonest Bob would not input $n$, since that will not tell him anything valuable about Alice's input.

The idea is that Alice and Bob want to see who is richer without revealing how much money they have. 
Ideally, Bob should not learn anything more about Alice's wealth other than that it is greater or less than some value. See Figure~\ref{fig_mp} for a pictorial representation of this application. 

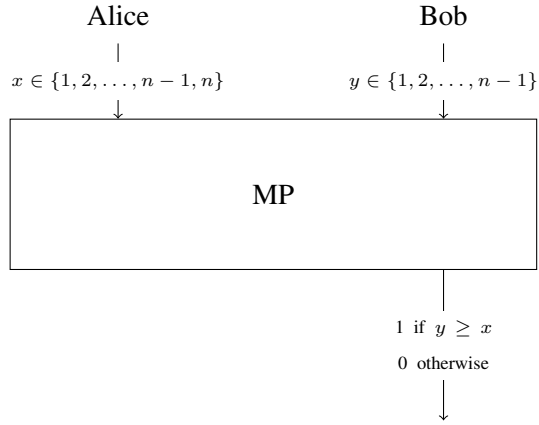
\begin{figure}[h]
\centering
\begin{tikzpicture}[node distance=1.11cm]
\node[draw,minimum width=7cm, minimum height=2cm] (MP) {MP};
\node (Alice) [above=of MP.154] {Alice};
\node (Bob) [above=of MP.24] {Bob};
\draw[->] ([yshift=1cm]MP.154) --  node[fill=white] {\ssmall $x \in \{1, 2, \ldots, n-1, n\}$} +(0pt,-1cm);
\draw[->] ([yshift=1cm]MP.24) -- node[fill=white] {\ssmall $y \in \{1, 2, \ldots, n-1\}$} +(0pt,-1cm);
\draw[->] (MP.336) -- node[fill=white, text width =2cm, align=center] {\ssmall 1 if $y\geq x$ 0 otherwise} +(0pt,-2cm);
\end{tikzpicture}
\caption{A pictorial representation of the millionaire's problem. Alice has a value $x \in \{1, 2, \ldots, n-1, n\}$ and Bob has a value $y \in \{1, 2, \ldots, n-1\}$. At the end of the protocol, Bob has learned if $y \geq x$.} 
\label{fig_mp}
\end{figure}

For each such protocol we define the following symbols.
\begin{center}
\begin{tabularx}{\textwidth}{rX}
  $\BMP$: & The maximum probability with which cheating Bob can guess honest Alice's input $x$. \\
  $\AMP$: & The maximum probability with which cheating Alice can guess honest Bob's input $y$. \\ 
\end{tabularx}
\end{center}
 
We can apply the lower bound from Inequality~\eqref{lb} noting that $\ASFE = \AMP$ and $\BSFEP = \BMP$. 
We can now apply the lower bound, and we get the following inequality    
\begin{equation} \label{lbMP}
\BMP \geq \frac{1}{\left(n-1\right) \AMP}-2\left(n-2\right)\sqrt{1-\frac{1}{\left(n-1\right)\AMP}}. 
\end{equation}  

We now argue that for this task that $\BSFER = \frac{2}{n}$. 
This is because for any input Bob may choose, if he sees that $f(x,y) = 0$ (which happens with probability $\frac{n-y}{n}$) then he has to randomly guess a number between $\{ y+1 , \ldots, n \}$ for his guess for $x$. 
But if he sees that $f(x,y) = 1$ (which happens with probability $\frac{y}{n}$) then he has to randomly guess a number between $\{ 1 , \ldots, y \}$ for his guess for $x$.  
Thus, 
\begin{equation} 
\BSFER = \frac{n-y}{n} \cdot \frac{1}{n-y} + \frac{y}{n} \cdot \frac{1}{y} = \frac{2}{n}.
\end{equation} 
Note that this makes sense as Bob learns $1$ bit of information about $x$ via the function $f$.  

Now we revisit the constants $c_A$ and $c_B$. We make the assumptions following Inequality~\eqref{ASFEc} and Equation~\eqref{BSFEc} that
\begin{equation}
\BMP = c_B \cdot \BSFER=\frac{2c_B}{n} \quad \text{and} \quad \AMP \leq \frac{c_A}{\left|Y\right|}=\frac{c_A}{n-1},\end{equation}
which, using Inequality~\eqref{lbMP}, provides us with a lower bound reminiscent of Inequality~\eqref{lb3}, i.e.,
\begin{equation}c_B \geq \frac{n}{2} \left(\frac{1}{c_A}-2\left(n-2\right)\sqrt{1-\frac{1}{c_A}}\right).
\end{equation}  
Figure \ref{MPc} illustrates the tradeoff between values of $c_B$ and $c_A$ for varying values of $n$. 

Using the rationale behind Equation~\eqref{exact}, we can introduce the constant $c$ such that
\begin{equation}\label{MPexact}c = \frac{n}{2} \left(\frac{1}{c}-2\left(n-2\right)\sqrt{1-\frac{1}{c}}\right).\end{equation}

\begin{figure}[h]
\centering
\scalebox{0.75}{%
\begin{tikzpicture}
\begin{axis}[ymin=1, ymax=2, xmin=1, xmax=1.026, axis lines=left, xlabel = {$c_A$}, ylabel={$c_B$}, width=16cm, height=10cm, x tick label style={/pgf/number format/.cd, precision=4, /tikz/.cd}]
\addplot+ [mark=none, thick, domain=1:1.059, samples=100] {(3*(1/(x)-2*(3-2)*sqrt(1-1/(x))))/2};
\addlegendentry{$n=3$}
\addplot+ [mark=none, thick, domain=1:1.027, samples=100] {(4*(1/(x)-2*(4-2)*sqrt(1-1/(x))))/2};
\addlegendentry{$n=4$}
\addplot+ [mark=none, thick, domain=1:1.0154, samples=100] {(5*(1/(x)-2*(5-2)*sqrt(1-1/(x))))/2};
\addlegendentry{$n=5$}
\addplot+ [mark=none, thick, domain=1:1.0099, samples=100] {(6*(1/(x)-2*(6-2)*sqrt(1-1/(x))))/2};
\addlegendentry{$n=6$} 
\end{axis}
\end{tikzpicture}%
}
\caption{$c_A$ vs. $c_B$ for the millionaire's problem and varying $n$.} 
\label{MPc}
\end{figure}
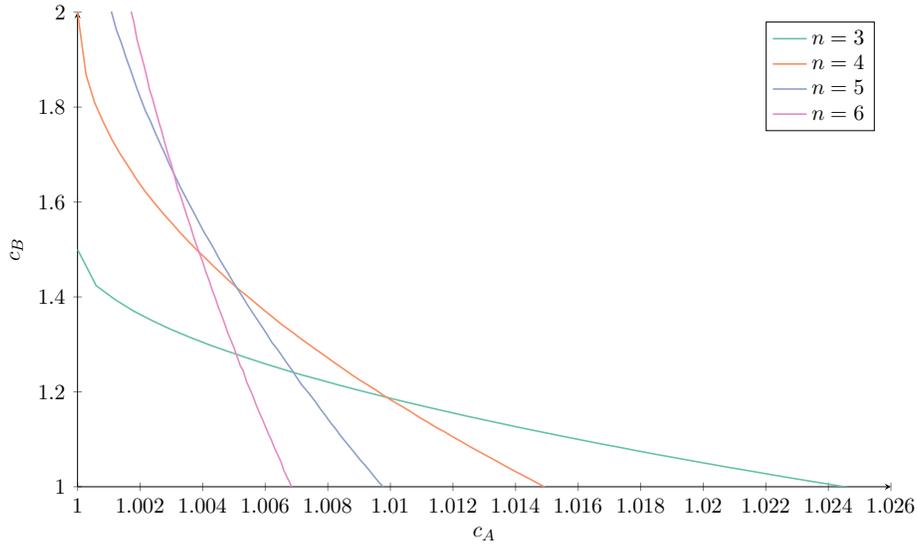


\subsubsection{Special cases} 
We now look at two special cases to see how the numbers behave. 

 
\paragraph{Millionaire's problem for $n=10^9$.} 
When $n=10^9$, we can calculate $c$ from~\eqref{MPexact} to get $c \approx 1+2.5 \times 10^{-19}$.  
This implies that 
\begin{equation} 
\begin{split}
\BMP &\gtrapprox 2 \times 10^{-9} + 5 \times 10^{-28}  > 2 \times 10^{-9}
\quad \text{ or } \quad \\
\AMP &\gtrapprox 1 \times 10^{-9} + 1 \times 10^{-18} + 1.25 \times 10^{-27} > 1 \times 10^{-9} + 1 \times 10^{-18} + 1 \times 10^{-27}.
\end{split}
\end{equation} 

Here we have capped the wealth of Alice and Bob at a billion dollars, but our bound works for any cap. 
 
 
\paragraph{Millionaire's problem, academics version.}  
When $n=10^1$, we can calculate $c$ from~\eqref{MPexact} to get $c \approx 1.0025$.  
This implies that 
\begin{equation} 
\BGS \gtrapprox 0.2005 > 0.2000
\quad \text{ or } \quad 
\AGS \gtrapprox 0.1114 > 0.1111.
\end{equation} 
Note that our lower bound seems to behave better for smaller values of $n$. 
Again, these are the only lower bounds for the cheating probabilities for the millionaire's problem of which we are aware. 
  
          
\section*{Acknowledgements} 

S.O. is supported by the Department of Defense Cyber Scholarship Program (DoD CySP).

\end{document}